\documentclass[11pt]{article}

\usepackage[american]{babel}
\usepackage[margin=1in]{geometry}
\usepackage[round]{natbib}
\renewcommand{\refname}{References}

\usepackage{mathtools}
\usepackage{amsmath,amssymb,amsfonts,amsthm}
\usepackage{booktabs}
\usepackage{graphicx}
\graphicspath{{./}{figures/}{figures_split_50_25_25/}}
\usepackage{enumitem}
\usepackage{xcolor}
\usepackage[most]{tcolorbox}
\usepackage{float}
\usepackage{algorithm}
\usepackage{algorithmic}
\usepackage{etoolbox}
\usepackage{subcaption}
\usepackage{adjustbox}
\usepackage{microtype}
\usepackage{hyperref}
\hypersetup{
  colorlinks=false,
  breaklinks=true,
  pdfborder={0 0 0.8},
  linkbordercolor={0 0 1},
  citebordercolor={1 1 1},
  urlbordercolor={1 1 1}
}

\AtBeginEnvironment{algorithmic}{\footnotesize}
\setlist[itemize]{leftmargin=*,topsep=2pt,itemsep=1pt,parsep=0pt}

\newcommand{\Ind}{\mathbb{I}}

\newcommand{\E}{\mathbb{E}}
\newcommand{\Prob}{\mathbb{P}}

\newtheorem{lemma}{Lemma}
\newtheorem{theorem}{Theorem}[section]
\newtheorem{proposition}[theorem]{Proposition}
\newtheorem{corollary}[theorem]{Corollary}

\newtheorem{assumption}[theorem]{Assumption}
\newtheorem{remark}[theorem]{Remark}

\begin{document}

\title{Detecting Unseen Jailbreak Sources: A Multi-Source Conformal Detection Perspective}
\author{
Xunlei Qian$^{1}$,
Yingqian Cui$^{2}$,
Yuping Lin$^{2}$,
Jun Wu$^{2}$,
Yue Xing$^{1}$\\[0.5em]
$^{1}$Department of Statistics and Probability, Michigan State University\\
$^{2}$Department of Computer Science and Engineering, Michigan State University
}

\date{}
\maketitle

\begin{abstract}
Jailbreak defenses for large language models are usually calibrated on a fixed set of known attack methods. However, new attack methods keep appearing after deployment, and existing detectors are observed to degrade under the new attacks \citep{piet2025jailbreaksovertime}. This paper studies how to detect, from a stream of attack prompts, that the prompts come from an attack source outside the known ones. We formulate it as a sequential test whose null hypothesis is that the stream comes from one of the known sources. Methodology-wise, we compute a conformal $p$-value against each known source, and reject the null only when the evidence against every known source is large. Using these source-wise $p$-values, we build Persist-3 and analyze its stationary false-alarm probability and  detection power. The technical contribution of Persist-3 lies in a detection guarantee that uses only a finite window of recent observations. Our theoretical analysis reveals that its detection probability approaches one exponentially fast as the window grows, and thus a short window suffices to detect unseen sources that are well separated from all known sources. Empirically, we validate Persist-3 on top of two recent jailbreak defenses, MTK \citep{zhang2026defending} and Prompt Guard 2 \citep{chennabasappa2025llamafirewall}, and observe a significant defense enhancement.

\end{abstract}

\section{INTRODUCTION}

Large language models (LLMs) have been deployed in many real applications, such as code generation \citep{chen2021evaluating}, interactive assistants \citep{ouyang2022training}, and agent systems that call external tools \citep{yao2022react}. Meanwhile, as these models take on more tasks, their safety becomes a practical concern.

In particular, LLMs are known to be vulnerable to jailbreak attacks \citep{zou2023universal,chao2025jailbreaking}, where a crafted prompt makes LLMs produce a harmful response. To defend against these attacks, most existing methods are developed and tested on standard benchmarks that contain a fixed collection of attack methods \citep{mazeika2024harmbench,chao2024jailbreakbench}. In practice, however, attackers keep designing new methods to bypass deployed defenses. For example, \citet{wang2025sokevaluatingjailbreakguardrails} evaluate existing jailbreak guardrails and find that their performance varies across attack types, and that generalizing to new attack strategies remains an open deployment challenge. Therefore, a defended system needs to know whether the incoming prompts are generated by one of the known attack methods.

To formulate the problem, we call an attack method applied to a given victim model an \emph{attack source}. Let $P_1,\dots,P_K$ denote the distributions of the victim-model hidden representations under the $K$ known sources, and let $Q_{t}$ denote the distribution generating the prompt at time $t$.
At each time $t$, we test $H_0: Q_{t} \in \{P_1,\dots,P_K\}$ against $H_1: Q_{t} \notin \{P_1,\dots,P_K\}$, i.e., instead of determining the specific source of the attack, we want to test whether the incoming attack is totally new compared to the known sources.

However, existing tools do not directly address this test. First, jailbreak detectors and guardrails judge each individual prompt separately. Thus, they do not track whether the stream of attack prompts, as a whole, has moved away from the known attack methods \citep{mazeika2024harmbench,chao2024jailbreakbench}. 

Second, change detectors \citep{shin2022detectors,saha2026non} usually target any departure from a common pre-change distribution. Thus, a switch between two known sources is also a change to be detected. Third, multi-distribution conformal prediction (MDCP) \citep{yang2026multidistributionrobustconformalprediction} calibrates on several sources at once.
However, it aims to build prediction sets that remain valid for every known source and their mixtures, and it does not test whether the current prompts come from a new source outside the known ones. These gaps lead to our question: 
\begin{center}
    \emph{From a stream of victim-model representations, how can we tell that no known source explains the current attack prompts, and what determines how quickly this can be detected?}
\end{center}
Technically, a central challenge of the considered test lies in the union structure of $H_0$. To see this, consider two known sources $P_1$ and $P_2$, and a stream from $P_1$. The $p$-value against $P_1$ is typically large, while the $p$-value against $P_2$ is small, since the stream is far from $P_2$. Standard $p$-value combinations, such as Bonferroni, reject as soon as one known source is rejected. Thus, they raise an alarm on this stream, though no unseen source appears. Different from Bonferroni, we reject $H_0$ only when every known source is rejected, i.e., when even the best-matching known source fails to explain the stream. 
We then accumulate this evidence over a recent window or over the whole history.

Our contributions are as follows:
\begin{itemize}
\item \textbf{Formulation.} We treat unseen jailbreak source detection as a sequential test whose null hypothesis is that the stream comes from one of the known sources. Compared to existing composite sequential tests \citep{saha2026non}, a switch between two known sources is not a detection target under our null (Section~\ref{sec:formulation}).
\item \textbf{Theory.} For our window-based detector, Persist-3, we provide theoretical justifications for its false alarms, memory, and detection power. First, we bound its false alarm probability under a known source (Theorem~\ref{thm:persist-null}). Second, we show that observations before a switch stop affecting its decision after a fixed number of steps (Proposition~\ref{prop:kk-recovery}). Third, we give a lower bound on its detection power controlled by the smallest evidence margin against the known sources (Theorem~\ref{thm:persist-power} and Corollary~\ref{cor:nearest-source}).
\item \textbf{Experiments.} Beyond the theory, we compare Persist-3, which uses only the most recent $W$ observations, with SR-DKW, which accumulates evidence over the whole history of the stream, and obtain two findings. First, the two detectors trade off against each other. SR-DKW detects more unseen sources (0.648 versus 0.458), while the decision of Persist-3 depends only on the most recent $W+L-1$ observations, so the old source stops affecting it shortly after a switch
(Section~\ref{sec:matched-regimes}). Second, adding either detector on top of MTK \citep{zhang2026defending} or Prompt Guard 2 \citep{chennabasappa2025llamafirewall} raises 7.0 to 10.0 percentage points more of the originally successful jailbreak attacks, i.e., escalates them before a harmful response is returned (Section~\ref{sec:defense}).
\end{itemize}

\section{Related Work}
\label{sec:related}

Our work connects conformal inference and jailbreak defense. On the conformal side, split conformal prediction gives distribution-free validity under exchangeability \citep{vovk2005algorithmic,angelopoulos2023conformal}. Multi-source methods retain this validity when calibration data come from several distributions \citep{liu2024multi,wong2025weightedpvalues,yang2026multidistributionrobustconformalprediction}. In addition, shift-adaptive and online conformal methods update thresholds when the test distribution changes \citep{TibshiraniBCR19,gibbs2021adaptive,zaffran2022adaptive}. These methods aim at valid prediction sets under a given shift model. Different from them, we use source-wise conformal $p$-values to test whether the stream is still generated by any calibrated source. On the jailbreak side, attack strategies evolve over time, and detectors degrade under such shifts \citep{piet2025jailbreaksovertime,wang2025sokevaluatingjailbreakguardrails}. Meanwhile, recent studies show that hidden representations separate harmful and benign prompts \citep{lin2024understandingjailbreakattacksllms,gao2025shaping}, which motivates our use of victim-model representations as the detection signal. Appendix~\ref{app:related} gives a detailed discussion.

\section{Method}
\label{sec:method}

This section introduces the proposed detector. We first formulate the problem (Section~\ref{sec:formulation}), then measure how atypical each incoming prompt is under each known source through a conformal $p$-value (Section~\ref{sec:source-p}). Finally we combine this evidence across sources and over time into a detection rule (Sections~\ref{sec:detectnew} to \ref{sec:calibration}), and describe how an alarm is used on top of an existing guardrail (Section~\ref{sec:post-detection}). Figure~\ref{fig:pipeline} in Appendix~\ref{app:algorithm} gives an overview of the pipeline.

\subsection{Problem formulation}
\label{sec:formulation}

\textbf{Data.} For each incoming prompt, we extract the hidden representation $X\in\mathcal X$ of the victim model before any response is returned. Let $Y\in\{0,1\}$ indicate whether the jailbreak succeeds. Let $P_1,\ldots,P_K$ denote the distributions of $X$ under the $K$ known attack sources, i.e., the sources whose data are available for calibration.

\textbf{Stream model and hypotheses.} Following the sequential change-detection literature \citep{shin2022detectors}, we model the stream as piecewise i.i.d. That is, the $X_t$ are independent with $X_t\sim Q_t$, and $Q_t$ stays constant between change points $\nu_1<\nu_2<\cdots$, at which the attacker switches source. We call the stretch between two change points a \emph{segment}. At each time $t$, we test $H_0:Q_t\in\{P_1,\ldots,P_K\}$ against $H_1:Q_t\notin\{P_1,\ldots,P_K\}$. Since the change points are unknown to the detector, it cannot pool all observations of the current segment. Thus, it uses only recent observations (Section~\ref{sec:detectnew}). We call a switch between two known sources a known-to-known (KK) switch, and a switch into an unseen source a known-to-unknown (KU) switch. A KK switch changes the distribution. However, since it stays within $H_0$, it is not a detection target. Only a KU switch is a detection target.

\textbf{Decisions after a switch.} Right after a switch, the recent window still contains prompts from the old source. Thus, our decisions are made after enough prompts from the new source have arrived. For Persist-3, this takes $B=W+L-1$ new observations, where the window length $W$ and the persistence length $L$ are defined in Sections~\ref{sec:detectnew} and \ref{sec:persistence}. After these $B$ observations, the old source no longer affects the decision (\emph{finite forgetting}, Section~\ref{sec:persistence}).

\textbf{Online information.} The online detector uses only $X_t$. Since the outcome $Y_t$ is unavailable when a prompt arrives, the detector does not access it. Outcome labels are used only offline. Specifically, we use them to evaluate interventions in Section~\ref{sec:defense}. In addition, under the additional assumptions of Corollary~\ref{cor:post_enrollment_mdcp}, we use them for optional enrollment of a detected source (Section~\ref{sec:post-detection}).

\subsection{Source-wise conformal $p$-values}
\label{sec:source-p}

This subsection converts each prompt into one $p$-value per known source. To decide whether a prompt is typical under source $k$, we first need a score measuring how far it is from that source. We use a Mahalanobis distance in a low-dimensional space. Specifically, we fit a shared PCA map $\Pi$ on the pooled training representations of all known sources, with output dimension at most $32$. This dimension is fixed in advance, and Appendix~\ref{app:pca-dim} shows that detection levels off at this dimension, while the KK alarm rate is insensitive to it. Then, for each source $k$, we fit a mean $\widehat\mu_k$ and a Ledoit-Wolf covariance $\widehat\Sigma_k$ in the projected space. Writing $z=\Pi(x)$, the nonconformity score is
\begin{equation}
S_k(x)
=
(z-\widehat\mu_k)^\top
\widehat\Sigma_k^{+}
(z-\widehat\mu_k),
\label{eq:mahal}
\end{equation}
where $+$ denotes the Moore-Penrose inverse when needed. A larger $S_k(x)$ means that $x$ is farther from source $k$.

We then convert this score into a $p$-value using held-out calibration data. Let $R_{k,i}=S_k(X_{k,i}^{\mathrm{cal}})$, $i=1,\ldots,n_k$, be the calibration scores of source $k$. For a new representation $x$, define
\begin{equation}
p_k(x)
=
\frac{
1+\sum_{i=1}^{n_k}
\Ind\{R_{k,i}\ge S_k(x)\}
}{
n_k+1
}.
\label{eq:source-p}
\end{equation}
Intuitively, $p_k(x)$ is essentially the fraction of calibration prompts from source $k$ that are at least as far from source $k$ as $x$. Thus, a small $p_k(x)$ is evidence against source $k$. To add this evidence over time, we transform it into $Z_k(x)=-\log p_k(x)$, which is large when $x$ is atypical under source $k$.

\subsection{Window-based evidence across known sources}
\label{sec:detectnew}

This subsection combines the source-wise evidence over a recent window and then across the known sources. Since a single small $p$-value can occur by chance even when the prompt comes from source $k$, we average the evidence over the most recent $W$ observations. In our jailbreak evaluation, we use $W=10$. With $m_t=\min\{W,t\}$, define, for each known source $k$,
\begin{equation}
C_{k,t}
=
\frac{1}{m_t}
\sum_{r=\max\{1,t-W+1\}}^{t}
-\log p_k(X_r).
\label{eq:source-window}
\end{equation}
Since $H_0$ holds as long as at least one known source explains the stream, we then take the minimum over the $K$ known sources:
\begin{equation}
C_t^{\mathrm{fam}}
=
\min_{k\in[K]}C_{k,t}.
\label{eq:family-stat}
\end{equation}
In other words, $C_t^{\mathrm{fam}}$ (``fam'' for the family of known sources) is the evidence against the known source that best explains the recent observations. Thus, it is small whenever at least one known source explains them well. Conversely, $C_t^{\mathrm{fam}}\ge\tau$ holds if and only if $C_{k,t}\ge\tau$ for every $k$. Therefore, exceeding the threshold $\tau$ means that every known source is rejected.

\textbf{Order of aggregation.} The order of the two aggregations in \eqref{eq:source-window} and \eqref{eq:family-stat}
matters. Since the minimum over sources is taken after the average over time, the same known source must explain the whole window. Thus, a window that mixes prompts from two known sources can be inconsistent with every single known source, even when each prompt in it is compatible with one of them. This is why a window that contains a switch between two known sources can raise an alarm (Remark~\ref{rem:kk-transient}). Once all observations used by the decision come from the new known source, the false-alarm bound of Theorem~\ref{thm:persist-null} applies again, provided that its threshold condition $\tau>\mu_{0,j}$ holds.

\textbf{A cumulative comparator.} For comparison, we also consider SR-DKW, a detector that accumulates evidence over the whole history of the stream. It applies a Dvoretzky--Kiefer--Wolfowitz (DKW) correction to the same source-wise $p$-values, and then feeds them into a Shiryaev--Roberts recursion (Appendix~\ref{app:auxiliary}). Since SR-DKW combines the sources at each step before accumulating over time, it does not require a single known source to explain a whole window. Our theoretical analysis focuses on Persist-3. Since its window is finite, we can characterize exactly when earlier observations stop affecting the current decision. The two detectors are compared empirically in Section~\ref{sec:experiments}.

\subsection{Persistent confirmation}
\label{sec:persistence}

This subsection adds a persistence rule that removes isolated threshold crossings. Since a single crossing of the threshold can still occur by chance, the detector raises an alarm only when $C_t^{\mathrm{fam}}$ stays above $\tau$ for $L$ consecutive steps. Let
\[
A_t(\tau)=\Ind\{C_t^{\mathrm{fam}}\ge\tau\}.
\]
For a persistence length $L$ and $t\ge L$, define
\begin{equation}
D_t
=
\Ind\!\left\{
A_t(\tau)=A_{t-1}(\tau)=\cdots=A_{t-L+1}(\tau)=1
\right\},
\label{eq:persist-decision}
\end{equation}
and set $D_t=0$ for $t<L$. Equivalently, $D_t=\Ind\{\min_{s=t-L+1,\ldots,t}C_s^{\mathrm{fam}}\ge\tau\}$. The detection time is $T_{\rm alarm}=\inf\{t\ge L:D_t=1\}$, with $\inf\varnothing=\infty$. In the analysis, we keep computing $D_t$ after the first alarm. This allows us to describe its behavior at any fixed time. In practice, Algorithm~\ref{alg:detectnew} stops at $T_{\rm alarm}$ and escalates the stream, and later values of $D_t$ do not undo this escalation.

\textbf{Choice of $L$.} For our jailbreak evaluation, we use $W=10$ and $L=3$. That is, the detector averages the evidence against each known source over the last $10$ observations, takes the minimum over the $K$ known sources, and raises an alarm only when this minimum stays above $\tau$ for three consecutive steps.
This is where the name Persist-3 comes from. We fix $L=3$ in advance by a pilot on known sources only. Specifically, $L=3$ is the smallest $L$ such that runs of at least $L$ consecutive threshold crossings are rare on known-source streams. Appendix~\ref{app:persist-length} reports $L\in\{1,2,3,4\}$ and shows a trade-off: a larger $L$ gives a higher detection rate but a longer delay. We require persistence for the following reason. Streams from a known source tend to produce isolated threshold crossings, whereas streams from an unseen source tend to produce consecutive ones. Thus, persistence removes most chance crossings while keeping sustained evidence. Appendix~\ref{app:algorithm} relates this rule to the $M$-out-of-$N$ rule in radar detection \citep{schwartz1956coincidence} and to runs rules for Shewhart charts \citep{champ1987exact}.

\textbf{Finite forgetting.} The decision $D_t$ has a finite memory. Since $D_t$ uses $L$ consecutive windows of length $W$, it depends only on the most recent $B=W+L-1$ observations. Thus, after a switch, observations from the old source stop affecting the decision once $B$ new observations have arrived. We call this property \emph{finite forgetting} and state it formally in Proposition~\ref{prop:kk-recovery}. For $W=10$ and $L=3$, this takes $B=12$ observations. Cumulative detectors such as SR-DKW do not have this property, since their internal state may keep reflecting pre-switch evidence.

\subsection{Threshold calibration using known sources only}
\label{sec:calibration}

This subsection selects the threshold $\tau$ using validation streams from known sources only. Specifically, no data or labels from the held-out unseen source are used. Let $\widehat{\mathrm{FPR}}_{\mathrm{val}}(\tau)$ denote the fraction of validation streams from a single known source that raise an alarm. Let $\widehat{\mathrm{KK}}_{\mathrm{val}}(\tau)$ denote the fraction of validation streams with a KK switch that raise an alarm. We choose
\[
\widehat{\tau}
=
\min\left\{
\tau\in\mathcal{T}:
\widehat{\mathrm{FPR}}_{\mathrm{val}}(\tau)\le \alpha_{\mathrm{stat}},
\;
\widehat{\mathrm{KK}}_{\mathrm{val}}(\tau)\le \alpha_{\mathrm{KK}}
\right\},
\]
where $\mathcal{T}$ is a prespecified threshold grid and $(\alpha_{\mathrm{stat}},\alpha_{\mathrm{KK}})$ are fixed before KU evaluation. That is, $\widehat\tau$ is the smallest threshold whose alarm rates on known-source validation streams stay within the two caps. The caps, the grid, and the selected thresholds of each protocol are listed in Appendix~\ref{app:algorithm}. When we repeat the evaluation on a second data split (Appendix~\ref{app:maintext-deferred}), the selection rule is kept fixed, and only $\widehat\tau$ is recalibrated. The detector architecture, $W$, $L$, the source score, and the aggregation rule are not retuned. Finally, each observation requires one score and one $p$-value per known source. In addition, Persist-3 stores only the most recent $W$ values per source.

\subsection{Layered intervention}
\label{sec:post-detection}

This subsection describes how an alarm is used to intervene on top of an existing guardrail. Let $G_t\in\{0,1\}$ indicate whether an unchanged prompt-level guardrail intervenes on prompt $t$. Let $M$ denote a source monitor, i.e., Persist-3 or SR-DKW, and let $T_{\rm act}^{(M)}$ be the time from which $M$ escalates every incoming prompt, e.g., by sending it to a stricter review. Write $A_t^{(M)}=\Ind\{t\ge T_{\rm act}^{(M)}\}$. The combined action is $I_t^{(G+M)}=G_t\vee A_t^{(M)}$. In other words, the guardrail keeps screening each prompt, and the monitor only decides the time from which every prompt in the stream is escalated. Escalation activates at the alarm itself, $T_{\rm act}^{(M)}=T_{\rm alarm}$ of Section~\ref{sec:persistence}, and the alarmed prompt is escalated (Appendix~\ref{app:defense}). An escalation is not undone when old observations leave the window. The reason is that finite forgetting (Proposition~\ref{prop:kk-recovery}) concerns the statistic $D_t$, whereas an escalation stays on once it is activated.

\section{Theoretical Results}
\label{sec:theory}

This section provides theoretical guarantees for Persist-3. Specifically, we study the validity of the source-wise $p$-values (Section~\ref{sec:theory-pvalue}), false alarm control under a known source (Section~\ref{sec:theory-null}), finite forgetting after a KK switch (Section~\ref{sec:theory-forgetting}), and detection power (Section~\ref{sec:theory-power}). 
Proofs are in Appendix~\ref{app:proofs}. Further results, including why persistence helps after calibration and a score-space identifiability limit, are in Appendix~\ref{app:additional-theory}.

\begin{assumption}[Frozen offline fit and i.i.d.\ segments]
\label{ass:monitoring}
Let $\mathcal T$ contain all training data and fitting randomness, so that the PCA map and all scores $S_k$ are fixed given $\mathcal T$. For each $k$, conditional on $\mathcal T$, the $n_k\ge1$ calibration observations are i.i.d.\ from $P_k$, with $n_k$ fixed before they are inspected. Independence across source calibration sets is not required.

Let $\mathcal O$ contain $\mathcal T$, all calibration and validation data, and the selected $W,L\ge1$ and $\tau$, all frozen before the final stream and independent of it. For deterministic change points $1=\nu_0<\nu_1<\cdots$, whose last endpoint may be $\infty$, and fixed distributions $Q_r$, each segment $\{X_t:\nu_r\le t<\nu_{r+1}\}$ is i.i.d.\ from $Q_r$ conditional on $\mathcal O$, and a known segment has $Q_r=P_j$ for some $j$.
\end{assumption}

Assumption~\ref{ass:monitoring} has two parts. The first part says that, once the offline fit is frozen, the calibration observations of each source behave like fresh draws from that source. Thus, it gives the exchangeability used by Proposition~\ref{prop:source_validity}. The second part says that the stream is a sequence of i.i.d.\ segments, and that all offline objects are chosen without looking at it. Thus, $Z_k(X_t)$ is independent and bounded within a segment, which is required by the weighted Hoeffding arguments in Theorems~\ref{thm:persist-null} and \ref{thm:persist-power}. Both parts follow standard practice, e.g., piecewise i.i.d.\ segments in sequential change detection \citep{shin2022detectors}, and split-conformal validity under a frozen score \citep{vovk2005algorithmic,angelopoulos2023conformal}.  Appendix~\ref{app:dwell-time} studies dwell times shorter than $B=W+L-1$.

\subsection{Source-wise conformal validity}
\label{sec:theory-pvalue}

\begin{proposition}[Source-wise split-conformal validity]
\label{prop:source_validity}
Under Assumption~\ref{ass:monitoring}, for a fresh $X\sim P_k$ independent of calibration given $\mathcal T$,
\[
\Prob\{p_k(X)\le u\mid\mathcal T\}\le u,\; u\in[0,1].
\]
Here, the probability is taken over both the calibration data and the new prompt $X$. The result also holds when the calibration scores and the score of $X$ are exchangeable, which is weaker than i.i.d. 
If two scores are equal with probability zero, i.e., there are no ties, then $p_k(X)$ is exactly uniform on $\{1/(n_k+1),\ldots,1\}$ conditional on $\mathcal T$.
\end{proposition}

The proof of Proposition~\ref{prop:source_validity} is in Appendix~\ref{app:proofs}. In short, it uses the standard rank argument of split conformal prediction. Note that Proposition~\ref{prop:source_validity} concerns a single observation. Since the same calibration sample is reused over time, it does not imply that the $p$-values along the stream are conditionally i.i.d.\ uniform. In addition, it does not give an anytime guarantee. A DKW argument further bounds the tail of $p_k$ conditional on the calibration data (Lemma~\ref{lem:cal-tail}, Appendix~\ref{app:additional-theory}). We use this bound only in the analysis, while the detector uses the raw $p$-values.

\subsection{False alarm control under a known source}
\label{sec:theory-null}

This subsection bounds the false alarm probability of Persist-3 under a known source. To explain the result, we first introduce some notation. Throughout this subsection, we condition on $\mathcal O$. Write $b_k=\log(n_k+1)$, so $0\le Z_k(X)\le b_k$. Let $\mu_{0,j}=\E_{P_j}[Z_j(X)\mid\mathcal O]$ be the mean evidence against a known source $P_j$ under $P_j$ itself. A Persist-$L$ decision uses $L$ overlapping windows of length $W$, which cover $B=W+L-1$ observations. Let $a_h$ be the number of windows containing the $h$th of these observations, and let $A_{W,L}=\sum_{h=1}^{B}a_h^2$. When $W\ge L$,
\begin{equation}
A_{W,L}
=
2\sum_{r=1}^{L-1}r^2
+(W-L+1)L^2.
\label{eq:awl}
\end{equation}

\begin{theorem}[Persistence-aware stationary false-alarm bound]
\label{thm:persist-null}
Under Assumption~\ref{ass:monitoring}, let $t\ge B$ be a deterministic time such that $X_{t-B+1},\ldots,X_t$ lie in one segment from $P_j$. All $L$ windows then have length $W$. If $\tau>\mu_{0,j}$, then
\begin{equation}
\Prob_{P_j}\{D_t=1\mid\mathcal O\}
\le
\exp\left\{
-\frac{2(LW)^2(\tau-\mu_{0,j})^2}
{b_j^2A_{W,L}}
\right\}.
\label{eq:persist-null-bound}
\end{equation}
\end{theorem}

The proof of Theorem~\ref{thm:persist-null} is in Appendix~\ref{app:proofs}. In short, a Persist-$L$ alarm forces all $L$ window averages against the true source $P_j$ above $\tau$. Their sum is a weighted sum of $B$ independent bounded terms with weights $a_h$. Then, the weighted Hoeffding inequality gives \eqref{eq:persist-null-bound}.

Theorem~\ref{thm:persist-null} illustrates two observations. First, the bound decays exponentially in the gap $\tau-\mu_{0,j}$. Thus, a threshold above the typical evidence under a known source gives rare false alarms. Second, the bound concerns the persistence event itself, not only the larger event $\{C_t^{\rm fam}\ge\tau\}$. Thus, it accounts for the dependence among the overlapping windows. For example, for $W=10$ and $L=3$, $A_{10,3}=82$, giving an effective sample size $(LW)^2/A_{W,L}\approx10.98$ instead of $LW=30$. In addition, Corollary~\ref{cor:finite-horizon} extends the bound to a finite horizon, including the first steps when the windows are not yet full, and Corollary~\ref{cor:dkw-null} gives a DKW refinement. Finally, since $\tau$ is selected from validation data, these bounds do not by themselves guarantee a prespecified population-level alarm rate for the selected threshold.

\subsection{Known-to-known switching and finite forgetting}
\label{sec:theory-forgetting}

\begin{proposition}[Finite forgetting within a known-source segment]
\label{prop:kk-recovery}
Suppose a KK switch occurs immediately before $X_\nu$, and the new segment follows $P_j$ until the next change point $\nu_{\rm next}$. Let $B=W+L-1$. At every deterministic time $\nu+B-1\le t<\nu_{\rm next}$, the continued decision $D_t$ depends only on observations from $P_j$, so Theorem~\ref{thm:persist-null} applies if $\tau>\mu_{0,j}$. 
\end{proposition}

The proof of Proposition~\ref{prop:kk-recovery} is in Appendix~\ref{app:proofs}. In other words, the old source stops affecting the decision after $B$ observations from the new source, which we call \emph{finite forgetting}. For $W=10$ and $L=3$, this takes $B=12$ observations. If the attacker switches again before $B$ new observations arrive, no fully post-switch decision exists, and Proposition~\ref{prop:kk-recovery} does not apply. Appendix~\ref{app:dwell-time} studies this case.

\begin{remark}[Transient KK alarms]
\label{rem:kk-transient}
Finite forgetting does not imply zero alarm probability right after a KK switch. 
During the first $B-1$ decisions after a switch from $P_i$ to $P_j$, a window contains prompts from both sources. The prompts from $P_j$ raise the average evidence against $P_i$, and the prompts from $P_i$ raise the average evidence against $P_j$. Thus, neither source may explain the whole window, and an alarm can occur. Finite forgetting also does not undo an earlier alarm, and segments shorter than $B$ may admit no decision based on a single source.
\end{remark}

Remark~\ref{rem:kk-transient} is consistent with the experiments, where the KK alarm rate is small but nonzero.

\subsection{Detection power and nearest-source separation}
\label{sec:theory-power}

This subsection lower bounds the detection power of Persist-3 under an unseen source. We again condition on $\mathcal O$. For a candidate unseen source $Q$, let $\mu_k(Q)=\E_Q[-\log p_k(X)\mid\mathcal O]$ be the mean evidence against known source $k$ under $Q$. Let $\Delta_k(Q)=\mu_k(Q)-\tau$ be its margin above the threshold.

\begin{theorem}[Detection power under separation from every known source]
\label{thm:persist-power}
Under Assumption~\ref{ass:monitoring}, suppose a segment $[\nu,\nu_{\rm next})$ follows $Q$ and $\Delta_k(Q)>0$ for every known source $k$. At any deterministic $t$ satisfying $\nu+B-1\le t<\nu_{\rm next}$, where $B=W+L-1$,
\begin{equation}
\Prob_Q\{D_t=1\mid\mathcal O\}
\ge
1-
L\sum_{k=1}^{K}
\exp\left(
-\frac{2W\Delta_k^2(Q)}{b_k^2}
\right).
\label{eq:persist-power}
\end{equation}
\end{theorem}

The proof of Theorem~\ref{thm:persist-power} is in Appendix~\ref{app:proofs}. In short, each window average against source $k$ concentrates around $\tau+\Delta_k(Q)$ by Hoeffding's inequality. Then, two union bounds, over the $K$ sources and over the $L$ persistence times, give \eqref{eq:persist-power}. 

Theorem~\ref{thm:persist-power} has two implications. First, it gives a sufficient condition for detection. If the mean evidence under $Q$ exceeds $\tau$ against every known source, the detection probability approaches one as $W$ grows. Second, it shows which source dominates the bound. Since each term of the bound is exponential in $W\Delta_k^2(Q)/b_k^2$, the bound is dominated by the smallest \emph{normalized} margin $\Delta_k(Q)/b_k$. That is, the margin is measured against the range $b_k=\log(n_k+1)$ of the evidence for that source. Note that an unseen source, i.e., $Q\notin\{P_1,\ldots,P_K\}$, does not necessarily satisfy this condition. In addition, the bound is nontrivial only when the subtracted sum is below one. Finally, when the source changes again before $B$ observations arrive, the window mixes several segments, and \eqref{eq:persist-power} need not apply. Appendix~\ref{app:dwell-time} shows that detection then degrades, most strongly when unseen segments are interrupted by known ones, but it does not vanish. The second implication leads to the following corollary.

\begin{corollary}[Nearest-source separation]
\label{cor:nearest-source}
Let $\Delta_*(Q)=\min_{k\in[K]}\Delta_k(Q)$, which we call the \emph{nearest-source separation} of $Q$, and $b_{\max}=\max_{k\in[K]}b_k$. 
Under the segment and time conditions of Theorem~\ref{thm:persist-power}, if $\Delta_*(Q)>0$, then
\begin{equation}
\Prob_Q\{D_t=1\mid\mathcal O\}
\ge
1-LK\exp\left(
-\frac{2W\Delta_*^2(Q)}{b_{\max}^2}
\right).
\label{eq:nearest-source-power}
\end{equation}
\end{corollary}

Corollary~\ref{cor:nearest-source} states the second implication in terms of the raw margins. Specifically, it bounds every $b_k$ by $b_{\max}$. Thus, the power bound is controlled by the smallest source-wise evidence margin, not by the average margin over known sources. Rearranging \eqref{eq:nearest-source-power}, the bound reaches $1-\eta$ once
\[
W\ \ge\ \frac{b_{\max}^2}{2\Delta_*^2(Q)}\log\frac{LK}{\eta},
\]
which is the scale on which the window has to grow as the margin shrinks. 

Appendix~\ref{app:additional-theory} further gives a delay bound (Corollary~\ref{cor:delay}) and an identifiability limit (Proposition~\ref{prop:identifiability}).

\section{Experiments}
\label{sec:experiments}

In this section, we empirically verify the analysis in Sections~\ref{sec:method} and \ref{sec:theory}. Our experiments aim to justify (1) the trade-off between window-based and cumulative detection under the same calibration, including the bounded memory of Persist-3 (Section~\ref{sec:matched-regimes}); (2) the role of nearest-source separation (Corollary~\ref{cor:nearest-source}, Section~\ref{sec:separation}); and (3) the incremental protection and cost of adding the detectors to existing guardrails (Section~\ref{sec:defense}).

\paragraph{Models, attacks, and protocol.}
We use AdvBench/Harmful Behaviors \citep{zou2023universal} and HarmBench \citep{mazeika2024harmbench} with four victim models: Llama2-7B-chat, Qwen2.5-7B-Instruct, Gemma-2-9B-it, and Llama-3.1-8B-Instruct. The eleven attack sources range from optimization-based suffixes \citep{zou2023universal,hu2024efficient} to automated black-box search \citep{chao2025jailbreaking,mehrotra2024tree,liu2025autodan} and human-written prompts \citep{shen2024doanythingnow}. In each configuration, one source is held out as the unseen source, and the others are known. Sections~\ref{sec:matched-regimes} and \ref{sec:separation} use 330 leave-one-source-out configurations over seeds $\{0,7,11,23,42\}$ with a 50/25/25 training, calibration, and test split (the \emph{original split}). Section~\ref{sec:defense} uses a separate 50-configuration replay (Appendix~\ref{app:defense}). We compare Persist-3 with SR-DKW and SR-RAW, which apply the Shiryaev--Roberts recursion to DKW-adjusted and raw source-wise $p$-values, respectively (Appendix~\ref{app:method-families}). Each SR method uses either the full history or a window of $W=10$.

\paragraph{Calibration on known sources only.}
The main comparison uses a \emph{long-horizon protocol}: each stream has horizon $H=1000$, and KK and KU streams switch at $\nu=501$. For each method, the threshold is the most aggressive one that keeps the stationary FPR $\le0.01$ and the KK alarm rate $\le0.05$ on known-source validation streams. All settings are frozen before KU evaluation, and all methods are tested on the same streams. Section~\ref{sec:separation} and the memory checks in Appendix~\ref{app:finite-forgetting-check} use a \emph{short-horizon protocol} with 40 observations per segment, and the threshold of Persist-3 is selected as in Section~\ref{sec:calibration}. Stress tests at other budgets are in Appendix~\ref{app:latest-operating-audits}.

\paragraph{Metrics.} 
We define that a \emph{stationary} stream comes from a single known source without any switch. We use the following three metrics for evaluation: (1) The stationary false positive rate (FPR) is the fraction of stationary streams with any alarm. (2) The KK alarm rate is the fraction of KK streams with any alarm. (3) The KU detection rate is the fraction of KU streams whose first alarm occurs after the switch, i.e., in $[\nu,H]$. A stream with an alarm before the switch is not restarted, so it counts as a miss. Since the first two metrics count false alarms, a lower stationary FPR and a lower KK alarm rate are better. In contrast, a higher KU detection rate is better. Confidence intervals come from a bootstrap that resamples held-out sources.

\subsection{Window-based versus cumulative detection}
\label{sec:matched-regimes}

This subsection shows that accumulating evidence over the whole history improves detection, while the window of Persist-3 bounds its memory. Table~\ref{tab:matched-main} gives three observations. 

First, full-history SR-DKW detects more unseen sources than Persist-3 under the same validation caps ($0.648$ vs.\ $0.458$), with a lower KK alarm rate ($0.003$ vs.\ $0.006$). 
In addition, both methods have a stationary FPR of $0/1835$. SR-RAW has an even higher KU detection rate ($0.720$). However, its stationary FPR is $8/1835$. Second, within each SR family, truncating the recursion to $W=10$ reduces KU detection (SR-DKW: $0.648\to0.377$; SR-RAW: $0.720\to0.545$). This is more direct evidence that accumulation helps, since weak but persistent evidence reaches the threshold only after it is accumulated. Third, Persist-3 has bounded memory, while full-history SR does not. Since each decision of Persist-3 uses only the most recent $B=W+L-1=12$ observations, the observations from the old source stop affecting its decision after $B$ post-switch steps (Proposition~\ref{prop:kk-recovery}). In contrast, full-history SR keeps the pre-switch evidence in its state (Remark~\ref{rem:sr-memory}). Appendix~\ref{app:finite-forgetting-check} verifies this bounded memory empirically.

\begin{table*}[t]
\centering
  \caption{Main comparison ($H=1000$, $\nu=501$). Thresholds are selected on known-source validation streams (stationary FPR $\le0.01$, KK alarm rate $\le0.05$). Stationary FPR is reported as counts.}
\label{tab:matched-main}
\footnotesize
\setlength{\tabcolsep}{5pt}
\begin{tabular}{lrrr}
\toprule
Method & Stationary FPR $\downarrow$ & KK alarm rate $\downarrow$ & KU detection rate $\uparrow$ \\
\midrule
Persist-3 & $0/1835$ & $0.006\,[0.004,0.007]$ & $0.458\,[0.335,0.627]$ \\
SR-DKW full & $0/1835$ & $0.003\,[0.002,0.005]$ & $0.648\,[0.490,0.793]$ \\
SR-RAW full & $8/1835$ & $0/10640$ & 0.720 \\
SR-DKW $W=10$ & $0/1835$ & 0.003 & 0.377 \\
SR-RAW $W=10$ & $4/1835$ & 0.005 & 0.545 \\
\bottomrule
\end{tabular}
\end{table*}

\subsection{Nearest-source separation and detectability}
\label{sec:separation}

This subsection shows that the nearest-source separation predicts which unseen sources are hard to detect. Corollary~\ref{cor:nearest-source} predicts that an unseen source is harder to detect when it is close to some known source. After freezing the detector, we estimate the separation of each held-out source by $\widehat\Delta_*(Q)=\min_{k\in[K]}\{\widehat\mu_k(Q)-\tau\}$, where $\widehat\mu_k(Q)$ is the average of $-\log p_k$ over the test prompts from $Q$. Across the eleven held-out attack methods, $\widehat\Delta_*$ is strongly associated with KU detection (Pearson $r=0.98$ [0.94, 1.00]; Spearman $\rho=0.95$ [0.70, 1.00]). The robustness resplit gives Pearson $r=0.981$ [0.947, 0.995]. Thus, attack methods with a small margin against at least one known source are hard to detect, which is consistent with Corollary~\ref{cor:nearest-source}. 

\begin{table*}[!tb]
\centering
\caption{Adding a source monitor to an unchanged guardrail (50 configurations, offline replay). Block rate: fraction of originally successful attacks in KU streams (including the known-source prefix) that are intercepted, with a 95\% bootstrap interval for its increment. Extra KU intervention: fraction of KU-stream prompts escalated by the monitor only. KK intervention: fraction of KK-stream prompts intervened on. Yield: additional interceptions per 100 extra KU interventions. Rates are averaged over configurations. Interception assumes that escalation prevents the harmful response. Details are in Appendix~\ref{app:defense}.}
\label{tab:layered}
\footnotesize
\setlength{\tabcolsep}{5pt}
\begin{tabular}{lcccc}
\toprule
System & Block rate (\%) $\uparrow$ & Extra KU intervention (\%) & KK intervention (\%) $\downarrow$ & Yield $\uparrow$ \\
\midrule
MTK alone & 39.7 & 0.0 & 62.2 & -- \\
\quad $+$ Persist-3 & 47.1 \;($+7.4\,[3.3,11.8]$) & 8.0 & 63.4 \;($+1.1$) & 7.28 \\
\quad $+$ SR-DKW & 49.7 \;($+10.0\,[4.1,16.1]$) & 10.9 & 62.7 \;($+0.4$) & 7.66 \\
\midrule
PG2 alone & 30.4 & 0.0 & 31.0 & -- \\
\quad $+$ Persist-3 & 37.4 \;($+7.0\,[2.6,12.6]$) & 12.1 & 32.7 \;($+1.7$) & 5.32 \\
\quad $+$ SR-DKW & 40.0 \;($+9.6\,[2.8,18.8]$) & 16.1 & 31.7 \;($+0.7$) & 5.83 \\
\bottomrule
\end{tabular}
\end{table*}

\subsection{Detecting new attack sources after guardrail failure}
\label{sec:defense}
This subsection shows that source detection adds protection on top of existing jailbreak defenses.
\paragraph{Setup.}
We treat known and unknown sources differently. 
\begin{itemize}
    \item \textbf{Known sources}: The system developer already knows how to defend against known sources. Thus, on streams with only known sources, we only report how often prompts are escalated.
    \item \textbf{Unknown source}: The considered methods are used to defend against the new source. Thus, we measure how many originally successful attacks in KU streams are blocked.
\end{itemize} We compare each unchanged guardrail $G$, MTK \citep{zhang2026defending} or Prompt Guard 2 (PG2) \citep{chennabasappa2025llamafirewall}, with $G+$Persist-3 and $G+$SR-DKW under the combined action $I_t^{(G+M)}=G_t\vee A_t^{(M)}$ of Section~\ref{sec:post-detection}. The evaluation is an offline replay of cached streams from 50 AdvBench configurations with Llama-2-7B-chat, which contain 1,729 originally successful attack occurrences.

\paragraph{Results.}
Table~\ref{tab:layered} shows three observations. First, adding either monitor increases the block rate for both guardrails. Specifically, Persist-3 increases the block rate by 7.4 and 7.0 percentage points for MTK and PG2, and SR-DKW increases it by 10.0 and 9.6 percentage points. All four bootstrap intervals are above zero. Second, the monitors intervene on more KU prompts. Specifically, the extra KU intervention is 8.0 to 16.1 percentage points, and the yield is 5.32 to 7.66. However, the KK intervention increases by only 0.4 to 1.7 percentage points. Third, SR-DKW blocks more attacks than Persist-3, which is consistent with Section~\ref{sec:matched-regimes}. Specifically, SR-DKW detects more KU streams ($0.851$ vs.\ $0.640$), and its block rate is 2.6 percentage points higher for both guardrails. 

\section{Conclusion}
This paper studies how to detect unseen jailbreak sources from a stream of attack prompts. We formulate the problem as a test of whether the stream comes from one of the known attack sources, and thus a switch among known sources stays within the null. For Persist-3, a window-based detector built on source-wise conformal $p$-values, we show that pre-switch observations stop affecting the decision after $W+L-1$ steps (finite forgetting). In addition, its detection power is controlled by the smallest evidence margin against the known sources (nearest-source separation). Empirically, experiments on jailbreak benchmarks verify these results. As a future direction, it is interesting to analyze the asymptotic behavior of the detector as more online data arrive.

\renewcommand{\refname}{References}
\bibliography{citation}

@article{Lei03072018,
  title={Distribution-free predictive inference for regression},
  author={Lei, Jing and G’Sell, Max and Rinaldo, Alessandro and Tibshirani, Ryan J and Wasserman, Larry},
  journal={Journal of the American Statistical Association},
  volume={113},
  number={523},
  pages={1094--1111},
  year={2018},
  publisher={Taylor \& Francis}
}

@inproceedings{shen2024doanythingnow,
  title={" do anything now": Characterizing and evaluating in-the-wild jailbreak prompts on large language models},
  author={Shen, Xinyue and Chen, Zeyuan and Backes, Michael and Shen, Yun and Zhang, Yang},
  booktitle={Proceedings of the 2024 on ACM SIGSAC Conference on Computer and Communications Security},
  pages={1671--1685},
  year={2024}
}

@article{liu2023autodan,
  title={Autodan: Generating stealthy jailbreak prompts on aligned large language models},
  author={Liu, Xiaogeng and Xu, Nan and Chen, Muhao and Xiao, Chaowei},
  journal={arXiv preprint arXiv:2310.04451},
  year={2023}
}

@article{lin2024understandingjailbreakattacksllms,
  title={Towards understanding jailbreak attacks in llms: A representation space analysis},
  author={Lin, Yuping and He, Pengfei and Xu, Han and Xing, Yue and Yamada, Makoto and Liu, Hui and Tang, Jiliang},
  journal={arXiv preprint arXiv:2406.10794},
  year={2024}
}

@article{ouyang2022training,
  title={Training language models to follow instructions with human feedback},
  author={Ouyang, Long and Wu, Jeffrey and Jiang, Xu and Almeida, Diogo and Wainwright, Carroll and Mishkin, Pamela and Zhang, Chong and Agarwal, Sandhini and Slama, Katarina and Ray, Alex and others},
  journal={Advances in neural information processing systems},
  volume={35},
  pages={27730--27744},
  year={2022}
}

@article{zou2023universal,
  title={Universal and transferable adversarial attacks on aligned language models},
  author={Zou, Andy and Wang, Zifan and Carlini, Nicholas and Nasr, Milad and Kolter, J Zico and Fredrikson, Matt},
  journal={arXiv preprint arXiv:2307.15043},
  year={2023}
}

@article{TibshiraniBCR19,
  title={Conformal prediction under covariate shift},
  author={Tibshirani, Ryan J and Foygel Barber, Rina and Candes, Emmanuel and Ramdas, Aaditya},
  journal={Advances in neural information processing systems},
  volume={32},
  year={2019}
}

@article{wang2025sokevaluatingjailbreakguardrails,
  title={SoK: Evaluating Jailbreak Guardrails for Large Language Models},
  author={Wang, Xunguang and Ji, Zhenlan and Wang, Wenxuan and Li, Zongjie and Wu, Daoyuan and Wang, Shuai},
  journal={arXiv preprint arXiv:2506.10597},
  year={2025}
}

@article{mazeika2024harmbench,
  title={Harmbench: A standardized evaluation framework for automated red teaming and robust refusal},
  author={Mazeika, Mantas and Phan, Long and Yin, Xuwang and Zou, Andy and Wang, Zifan and Mu, Norman and Sakhaee, Elham and Li, Nathaniel and Basart, Steven and Li, Bo and others},
  journal={arXiv preprint arXiv:2402.04249},
  year={2024}
}

@article{chao2024jailbreakbench,
  title={Jailbreakbench: An open robustness benchmark for jailbreaking large language models},
  author={Chao, Patrick and Debenedetti, Edoardo and Robey, Alexander and Andriushchenko, Maksym and Croce, Francesco and Sehwag, Vikash and Dobriban, Edgar and Flammarion, Nicolas and Pappas, George J and Tramer, Florian and others},
  journal={Advances in Neural Information Processing Systems},
  volume={37},
  pages={55005--55029},
  year={2024}
}

@book{vovk2005algorithmic,
  title={Algorithmic learning in a random world},
  author={Vovk, Vladimir and Gammerman, Alexander and Shafer, Glenn},
  year={2005},
  publisher={Springer}
}

@article{shafer2008tutorial,
  title={A tutorial on conformal prediction.},
  author={Shafer, Glenn and Vovk, Vladimir},
  journal={Journal of Machine Learning Research},
  volume={9},
  number={3},
  year={2008}
}

@article{wang2025conformalpredictionadaptiveunknown,
  title={Conformal Prediction Adaptive to Unknown Subpopulation Shifts},
  author={Wang, Nien-Shao and Yaldiz, Duygu Nur and Bakman, Yavuz Faruk and Karimireddy, Sai Praneeth},
  journal={arXiv preprint arXiv:2506.05583},
  year={2025}
}

@article{wong2025weightedpvalues,
  title={Improving Coverage in Combined Prediction Sets with Weighted p-values},
  author={Wong, Gina and Prinster, Drew and Saria, Suchi and Chellappa, Rama and Liu, Anqi},
  journal={arXiv preprint arXiv:2505.11785},
  year={2025}
}

@article{pmlr-v235-angelopoulos24a,
  title={Online conformal prediction with decaying step sizes},
  author={Angelopoulos, Anastasios N and Barber, Rina Foygel and Bates, Stephen},
  journal={arXiv preprint arXiv:2402.01139},
  year={2024}
}

@article{gibbs2021adaptive,
  title={Adaptive conformal inference under distribution shift},
  author={Gibbs, Isaac and Candes, Emmanuel},
  journal={Advances in Neural Information Processing Systems},
  volume={34},
  pages={1660--1672},
  year={2021}
}

@article{fayyazi25a,
  title={Facter: Fairness-aware conformal thresholding and prompt engineering for enabling fair llm-based recommender systems},
  author={Fayyazi, Arya and Kamal, Mehdi and Pedram, Massoud},
  journal={arXiv preprint arXiv:2502.02966},
  year={2025}
}

@article{yang2026multidistributionrobustconformalprediction,
  title={Multi-Distribution Robust Conformal Prediction},
  author={Yang, Yuqi and Jin, Ying},
  journal={arXiv preprint arXiv:2601.02998},
  year={2026}
}

@article{angelopoulos2023conformal,
  title={Conformal prediction: A gentle introduction},
  author={Angelopoulos, Anastasios N and Bates, Stephen},
  journal={Foundations and Trends in Machine Learning},
  volume={16},
  number={4},
  pages={494--591},
  year={2023},
  publisher={Emerald Publishing Limited}
}

@article{hu2024efficient,
  title={Efficient llm jailbreak via adaptive dense-to-sparse constrained optimization},
  author={Hu, Kai and Yu, Weichen and Li, Yining and Yao, Tianjun and Li, Xiang and Liu, Wenhe and Yu, Lijun and Shen, Zhiqiang and Chen, Kai and Fredrikson, Matt},
  journal={Advances in Neural Information Processing Systems},
  volume={37},
  pages={23224--23245},
  year={2024}
}

@article{takemoto2024all,
  title={All in how you ask for it: Simple black-box method for jailbreak attacks},
  author={Takemoto, Kazuhiro},
  journal={Applied Sciences},
  volume={14},
  number={9},
  pages={3558},
  year={2024},
  publisher={MDPI}
}

@inproceedings{chao2025jailbreaking,
  title={Jailbreaking black box large language models in twenty queries},
  author={Chao, Patrick and Robey, Alexander and Dobriban, Edgar and Hassani, Hamed and Pappas, George J and Wong, Eric},
  booktitle={2025 IEEE Conference on Secure and Trustworthy Machine Learning (SaTML)},
  pages={23--42},
  year={2025},
  organization={IEEE}
}

@article{mehrotra2024tree,
  title={Tree of attacks: Jailbreaking black-box llms automatically},
  author={Mehrotra, Anay and Zampetakis, Manolis and Kassianik, Paul and Nelson, Blaine and Anderson, Hyrum and Singer, Yaron and Karbasi, Amin},
  journal={Advances in Neural Information Processing Systems},
  volume={37},
  pages={61065--61105},
  year={2024}
}

@inproceedings{piet2025jailbreaksovertime,
  title={Jailbreaksovertime: Detecting jailbreak attacks under distribution shift},
  author={Piet, Julien and Huang, Xiao and Jacob, Dennis and Chow, Annabella and Alrashed, Maha and Zhao, Geng and Hu, Zhanhao and Sitawarin, Chawin and Alomair, Basel and Wagner, David},
  booktitle={Proceedings of the 18th ACM Workshop on Artificial Intelligence and Security},
  pages={230--241},
  year={2025}
}

@inproceedings{gao2025shaping,
  title={Shaping the safety boundaries: Understanding and defending against jailbreaks in large language models},
  author={Gao, Lang and Geng, Jiahui and Zhang, Xiangliang and Nakov, Preslav and Chen, Xiuying},
  booktitle={Proceedings of the 63rd Annual Meeting of the Association for Computational Linguistics (Volume 1: Long Papers)},
  pages={25378--25398},
  year={2025}
}

@article{he2024jailbreaklens,
  title={Jailbreaklens: Interpreting jailbreak mechanism in the lens of representation and circuit},
  author={He, Zeqing and Wang, Zhibo and Chu, Zhixuan and Xu, Huiyu and Zhang, Wenhui and Wang, Qinglong and Zheng, Rui},
  journal={arXiv preprint arXiv:2411.11114},
  year={2024}
}

@article{zheng2024prompt,
  title={On prompt-driven safeguarding for large language models},
  author={Zheng, Chujie and Yin, Fan and Zhou, Hao and Meng, Fandong and Zhou, Jie and Chang, Kai-Wei and Huang, Minlie and Peng, Nanyun},
  journal={arXiv preprint arXiv:2401.18018},
  year={2024}
}

@article{chen2021evaluating,
  title={Evaluating large language models trained on code},
  author={Chen, Mark},
  journal={arXiv preprint arXiv:2107.03374},
  year={2021}
}

@article{yao2022react,
  title={React: Synergizing reasoning and acting in language models},
  author={Yao, Shunyu and Zhao, Jeffrey and Yu, Dian and Du, Nan and Shafran, Izhak and Narasimhan, Karthik and Cao, Yuan},
  journal={arXiv preprint arXiv:2210.03629},
  year={2022}
}

@inproceedings{spjuth2019combining,
  title={Combining prediction intervals on multi-source non-disclosed regression datasets},
  author={Spjuth, Ola and Br{\"a}nnstr{\"o}m, Robin Carri{\'o}n and Carlsson, Lars and Gauraha, Niharika},
  booktitle={Conformal and Probabilistic Prediction and Applications},
  pages={53--65},
  year={2019},
  organization={PMLR}
}

@article{liu2024multi,
  title={Multi-source conformal inference under distribution shift},
  author={Liu, Yi and Levis, Alexander W and Normand, Sharon-Lise and Han, Larry},
  journal={Proceedings of machine learning research},
  volume={235},
  pages={31344},
  year={2024}
}

@inproceedings{zaffran2022adaptive,
  title={Adaptive conformal predictions for time series},
  author={Zaffran, Margaux and F{\'e}ron, Olivier and Goude, Yannig and Josse, Julie and Dieuleveut, Aymeric},
  booktitle={International Conference on Machine Learning},
  pages={25834--25866},
  year={2022},
  organization={PMLR}
}

@article{shin2022detectors,
  title={E-detectors: A nonparametric framework for sequential change detection},
  author={Shin, Jaehyeok and Ramdas, Aaditya and Rinaldo, Alessandro},
  journal={arXiv preprint arXiv:2203.03532},
  year={2022}
}

@article{saha2026non,
  title={Non-partitioned e-detectors for nonparametric sequential change detection},
  author={Saha, Aytijhya and Ramdas, Aaditya},
  journal={arXiv preprint arXiv:2607.28322},
  year={2026}
}

@article{liu2024flipattack,
  title={Flipattack: Jailbreak llms via flipping},
  author={Liu, Yue and He, Xiaoxin and Xiong, Miao and Fu, Jinlan and Deng, Shumin and Ma, Yingwei and Zhang, Jiaheng and Hooi, Bryan},
  journal={arXiv preprint arXiv:2410.02832},
  year={2024}
}

@inproceedings{liu2025autodan,
  title={Autodan-turbo: A lifelong agent for strategy self-exploration to jailbreak llms},
  author={Liu, Xiaogeng and Li, Peiran and Suh, G Edward and Vorobeychik, Yevgeniy and Mao, Zhuoqing and Jha, Somesh and McDaniel, Patrick and Sun, Huan and Li, Bo and Xiao, Chaowei},
  booktitle={International Conference on Learning Representations},
  volume={2025},
  pages={10313--10360},
  year={2025}
}

@inproceedings{fu2026steering,
  title={Steering Away from Refusal: A Black-box Jailbreak Method Based on First-Token Distribution},
  author={Fu, Shuangjie and Su, Du and Chen, Xin and Sun, Fei and Shen, Huawei and Cheng, Xueqi},
  booktitle={Findings of the Association for Computational Linguistics: ACL 2026},
  pages={25969--25979},
  year={2026}
}

@inproceedings{hu2026jailbreak,
  title={Jailbreak-Zero: A Path to Pareto Optimal Red Teaming for Large Language Models},
  author={Hu, Kai and Aggarwal, Abhinav and Khodabandeh, Mehran and Zhang, David and Hsin, Eric and Chen, Li and Jain, Ankit and Fredrikson, Matt and Bharadwaj, Akash},
  booktitle={Proceedings of the 64th Annual Meeting of the Association for Computational Linguistics (Volume 1: Long Papers)},
  pages={46720--46746},
  year={2026}
}

@article{zhang2026defending,
  title={Defending Jailbreak Attacks on Large Language Models via Manifold Trajectory Kinetics},
  author={Zhang, Hangtao and Zhao, Yucheng and Liu, Sishun and Zhou, Ziqi and Ye, Zeyu and Wan, Wei and Li, Minghui and Hu, Shengshan and Zhang, Yanjun and Liu, Yi and others},
  journal={arXiv preprint arXiv:2606.07335},
  year={2026}
}

@article{chennabasappa2025llamafirewall,
  title={Llamafirewall: An open source guardrail system for building secure ai agents},
  author={Chennabasappa, Sahana and Nikolaidis, Cyrus and Song, Daniel and Molnar, David and Ding, Stephanie and Wan, Shengye and Whitman, Spencer and Deason, Lauren and Doucette, Nicholas and Montilla, Abraham and others},
  journal={arXiv preprint arXiv:2505.03574},
  year={2025}
}

@article{schwartz1956coincidence,
  author  = {Schwartz, Mischa},
  title   = {A Coincidence Procedure for Signal Detection},
  journal = {IRE Transactions on Information Theory},
  volume  = {2},
  number  = {4},
  pages   = {135--139},
  year    = {1956},
  doi     = {10.1109/TIT.1956.1056820}
}

@article{champ1987exact,
  author  = {Champ, Charles W. and Woodall, William H.},
  title   = {Exact Results for {Shewhart} Control Charts with Supplementary Runs Rules},
  journal = {Technometrics},
  volume  = {29},
  number  = {4},
  pages   = {393--399},
  year    = {1987},
  doi     = {10.1080/00401706.1987.10488266}
}

\clearpage
\appendix
\section*{Supplementary Materials}
\section{Robustness and Diagnostics Deferred from the Main Text}
\label{app:maintext-deferred}

\subsection{Robustness resplit}
We use two data splits throughout and never pool them. The \emph{original split} is a 50/25/25 partition over seeds $\{0,7,11,23,42\}$, giving 330 complete configurations, with the frozen Persist-3 threshold $\widehat\tau=2.386$. The \emph{robustness resplit} is a second 50/25/25 partition of the same cached example pool over seeds $\{101,103,107,109,113\}$, with the threshold recalibrated by the same known-only rule, $\widehat\tau=2.374$. The robustness resplit is a split-robustness check on the same cached embedding pool. It is not an independent dataset, not a preregistered confirmatory study, and not a second population of jailbreaks. Its protocol was fixed before KU results were computed, and the detector was not retuned. It gives 10/2425 stationary alarms (FPR $0.0041$), a KK alarm rate of $0.044$, and KU detection $0.414$ $[0.282,0.589]$ (Table~\ref{tab:pseudo-confirmatory-app}).

\subsection{Summary of theory--experiment diagnostics}
\label{sec:theory-empirical-alignment}
Appendix~\ref{app:similarity} reports diagnostics computed after the detector is frozen; they are not used for tuning. They show that (i) the estimated nearest-source separation is strongly associated with KU detection, as predicted by Corollary~\ref{cor:nearest-source}; (ii) two streams that differ only before a switch produce identical Persist-3 decisions after $B=12$ steps, and all first post-switch KK alarms occur within $d\le11$ steps, as predicted by Proposition~\ref{prop:kk-recovery} (Appendix~\ref{app:kk-timing}); and (iii) delay conditional on detection is biased by selection, so we report it only descriptively.

\section{Additional Theoretical Results}
\label{app:additional-theory}
This section collects results referred to in the main text. Proofs are in Appendix~\ref{app:proofs}.

\begin{lemma}[Calibration-conditional tail bound]
\label{lem:cal-tail}
Under Assumption~\ref{ass:monitoring}, fix $\mathcal T$ and $\beta\in(0,1)$. Let $F_k(s)=\Prob_{P_k}\{S_k(X)\le s\mid\mathcal T\}$ and let $\widehat F_k$ be the empirical CDF of the calibration scores. Define
\[
\varepsilon_k(\beta)=\sqrt{\frac{\log(2K/\beta)}{2n_k}},\;
\mathcal G=\bigcap_{k=1}^K\left\{\sup_s|\widehat F_k(s)-F_k(s)|\le\varepsilon_k(\beta)\right\}.
\]
Then $\Prob(\mathcal G\mid\mathcal T)\ge1-\beta$. For every offline realization $\mathcal O$ in $\mathcal G$ and a fresh $X\sim P_k$ independent of $\mathcal O$,
\[
\Prob_{P_k}\{p_k(X)\le u\mid\mathcal O\}\le\min\{1,u+\varepsilon_k(\beta)\},\; u\in[0,1].
\]
Consequently, simultaneously for all $k$,
\[
\E_{P_k}[-\log p_k(X)\mid\mathcal O]\le\bar\mu_k
:=\int_0^{b_k}\min\{1,e^{-z}+\varepsilon_k(\beta)\}\,dz,
\; b_k=\log(n_k+1).
\]
The probability $1-\beta$ averages over calibration, conditional on training; it is not a coverage guarantee conditional on calibration-dependent selection events.
\end{lemma}

\begin{corollary}[High-probability DKW version]
\label{cor:dkw-null}
On $\mathcal G$ in Lemma~\ref{lem:cal-tail}, $\mu_{0,j}\le\bar\mu_j$. Under the full-window conditions of Theorem~\ref{thm:persist-null}, whenever $\tau>\bar\mu_j$, its bound holds with $\mu_{0,j}$ replaced by $\bar\mu_j$. The joint calibration event has probability at least $1-\beta$ conditional on $\mathcal T$. This statement does not assert that an empirically selected $\tau$ satisfies $\tau>\bar\mu_j$.
\end{corollary}

\begin{corollary}[Finite-horizon stationary control including startup]
\label{cor:finite-horizon}
Under Assumption~\ref{ass:monitoring}, suppose $X_1,\ldots,X_H$ are i.i.d.\ from $P_j$, where $H\ge L$ is deterministic. For $L\le t\le H$, let
\[
v_{t,r}=\sum_{s=t-L+1}^{t}\frac{\Ind\{\max(1,s-W+1)\le r\le s\}}{m_s},
\; r=1,\ldots,t,\; V_t=\sum_{r=1}^t v_{t,r}^2.
\]
If $\tau>\mu_{0,j}$, define
\[
\rho_{j,t}=\exp\!\left\{-\frac{2L^2(\tau-\mu_{0,j})^2}{b_j^2V_t}\right\};
\]
otherwise set $\rho_{j,t}=1$. Then
\[
\Prob_{P_j}\{T_{\rm alarm}\le H\mid\mathcal O\}
\le\min\left\{1,\sum_{t=L}^{H}\rho_{j,t}\right\}.
\]
For $t\ge B$, $V_t=A_{W,L}/W^2$, recovering the full-window bound. On $\mathcal G$, $\mu_{0,j}$ may be replaced by $\bar\mu_j$ when $\tau>\bar\mu_j$, with the trivial bound used otherwise. If that DKW-based bound is at most $\alpha$ on every offline realization in $\mathcal G$, then the unconditional probability of any alarm in the stream is at most $\min\{1,\alpha+\beta\}$. Selecting the threshold on known-source validation data alone does not establish this premise.
\end{corollary}

\begin{proposition}[Persistence enlarges the validation-feasible threshold set]\label{prop:threshold-relax}
For every stream and deterministic threshold, a Persist-$L$ alarm implies a single-crossing alarm. Hence on any fixed collection of stationary or KK validation streams, the empirical Persist-$L$ alarm rate is no larger than the single-crossing rate at the same threshold; under common alarm caps, the set of feasible thresholds for Persist-$L$ contains that of single crossing.
\end{proposition}

\paragraph{Why persistence helps after calibration.}
At any fixed threshold, Persist-$L$ is at least as conservative as single crossing. Its empirical gain comes from \emph{recalibration}: by removing short excursions under known sources, persistence enlarges the set of thresholds satisfying the validation caps (Proposition~\ref{prop:threshold-relax}), which allows a lower threshold. This is consistent with the short-horizon comparison under common caps, where KU detection increases from $0.3755$ for single crossing to $0.4663$ for Persist-3 (Table~\ref{tab:short-horizon}).

\begin{proposition}[Score-space identifiability limit]
\label{prop:identifiability}
Fix $\mathcal O$ and compare i.i.d.\ streams from $Q$ and $P_j$. If the one-observation vector $(p_1(X),\ldots,p_K(X))$ has the same law under both sources, then any detector using only these score vectors, with the same fixed initial state, has the same alarm-time law. The statement also holds conditional on a common pre-change history when each future stream is i.i.d.\ and independent of that history.
\end{proposition}

That is, an unseen source that induces the same distribution of $p$-values as a known source cannot be detected by any detector based on these $p$-values, whatever its memory. In particular, if $Q$ and $P_j$ induce the same distribution of $\Pi(X)$, then every score and $p$-value computed from this projection has the same distribution under the two sources. Thus, changes confined to the discarded directions can be invisible to the detector. This limits power. However, it does not affect the validity of the source-wise conformal calibration.

\begin{corollary}[Post-change delay on disjoint blocks]
\label{cor:delay}
Under Theorem~\ref{thm:persist-power}, suppose $Q$ persists for $H$ observations starting at $\nu$, where $H\in\mathbb N\cup\{\infty\}$. Let
\[
q=L\sum_{k=1}^K\exp\{-2W\Delta_k^2(Q)/b_k^2\}<1,\; B=W+L-1,
\]
and define the delay to the first fully post-change alarm of the continued statistic $D_t$ (Section~\ref{sec:persistence}) by
\[
T_{\nu,H}=\inf\{d\in\mathbb N:B\le d\le H,\ D_{\nu+d-1}=1\},
\; \inf\varnothing=\infty.
\]
For each integer $m\ge0$ with $mB\le H$,
\[
\Prob_Q\{T_{\nu,H}>mB\mid\mathcal O\}\le q^m.
\]
For finite $H$, this gives $\Prob_Q\{T_{\nu,H}=\infty\mid\mathcal O\}\le q^{\lfloor H/B\rfloor}$. If $H=\infty$, it also gives
\[
\E_Q[T_{\nu,\infty}\mid\mathcal O]\le\frac{B}{1-q}.
\]
These bounds concern the continued statistic, not delay conditional on successful detection within a finite horizon or on the detector not having stopped before $\nu$.
\end{corollary}

\begin{corollary}[Post-enrollment coverage under conditional exchangeability]
\label{cor:post_enrollment_mdcp}
Let $\mathcal H$ specify the enrollment history, the enlarged set of $K+1$ joint laws of $(X,Y)$, and fixed outcome-prediction scores $s_k$. For each source $k$, suppose its labeled calibration block of size $N_k$ and a future labeled observation from that source are exchangeable conditional on $\mathcal H$. The block of the new source is collected after the alarm and does not reuse the observations that triggered it. Define
\[
\pi_k(x,y)=\frac{1+\sum_{i=1}^{N_k}\Ind\{s_k(X_{k,i},Y_{k,i})\ge s_k(x,y)\}}{N_k+1},
\;
\Gamma_\alpha(x)=\{y:\max_{k\le K+1}\pi_k(x,y)>\alpha\}.
\]
For every $j\le K+1$ and $\alpha\in(0,1)$,
\[
\Prob_j\{Y\in\Gamma_\alpha(X)\mid\mathcal H\}\ge1-\alpha.
\]
This is the coverage part of MDCP \citep{yang2026multidistributionrobustconformalprediction}, not an efficiency claim. Conditional exchangeability must hold for every retained source as well as the new one, and it does not follow merely from collecting new-source data after detection. Fresh independent calibration for each source after fixing its prediction score is sufficient; reusing older calibration data requires a separate justification.
\end{corollary}

No validated post-alarm enrollment evaluation is included in the empirical claims of this paper. Corollary~\ref{cor:post_enrollment_mdcp} describes a conditional downstream use case, not an experimentally demonstrated coverage-recovery procedure.

\section{Related Work}
\label{app:related}

\paragraph{Conformal prediction.}
Conformal prediction provides a general framework for constructing prediction sets with finite-sample, distribution-free validity under exchangeability \citep{vovk2005algorithmic,shafer2008tutorial}.
In split conformal prediction, the available data are divided into training and calibration subsets. A predictive model or nonconformity score is fitted using the training data, while the held-out calibration data determine the finite-sample prediction threshold. For a target miscoverage level $\alpha\in(0,1)$, this construction provides marginal coverage of at least $1-\alpha$ without requiring a correctly specified predictive model \citep{Lei03072018,angelopoulos2023conformal}.
This model-agnostic property is particularly useful when the underlying predictor is a complex neural network or large language model.

Our detection problem uses conformal calibration for a different purpose. Rather than constructing only a prediction set for a downstream label, we first convert source-wise scores into calibrated source-wise $p$-values, which measure how typical a prompt is under each known source. A separate set of conformal quantities is reserved for downstream prediction of the jailbreak outcome (Corollary~\ref{cor:post_enrollment_mdcp}). This distinction between calibration for detecting unseen sources and calibration for outcome prediction is central to the proposed framework.

\paragraph{Conformal prediction across multiple sources.}
Classical conformal prediction assumes that calibration and future observations are exchangeable. This assumption becomes more restrictive when data originate from heterogeneous populations or sources. Multi-source conformal methods study how prediction validity can be retained when calibration data arise from multiple distributions \citep{liu2024multi,spjuth2019combining}. One class of approaches calibrates separately for each source and aggregates the resulting source-wise evidence. Weighted combinations of conformal $p$-values can improve the efficiency of such combined prediction sets \citep{wong2025weightedpvalues}.

Multi-distribution robust conformal prediction (MDCP) further considers prediction sets that provide coverage simultaneously across a collection of calibrated distributions \citep{yang2026multidistributionrobustconformalprediction}. Its guarantee covers mixtures of the calibrated sources as well as each source separately, so a deployment distribution that equals no single $P_k$ may still lie inside its coverage guarantee.

Therefore, our question is not whether the coverage of MDCP still holds. Instead, we ask whether the current segment is generated by a single calibrated source, that is, whether
\[
Q\notin\{P_1,\ldots,P_K\},
\]
which is the quantity a defender needs in order to decide whether the calibrated source family still describes the incoming traffic. An alarm signals that this membership is implausible; it does not by itself establish that downstream prediction sets have lost coverage. If the detected source is later enrolled as $P_{K+1}$ with a fresh calibration block, MDCP applies to the enlarged family under the conditions of Corollary~\ref{cor:post_enrollment_mdcp}.

\paragraph{Conformal prediction under distribution shift.}
A complementary line of work studies conformal inference when the test distribution differs from the calibration distribution. Covariate-shift methods reweight calibration observations when the density ratio between training and test distributions can be estimated \citep{TibshiraniBCR19}. More recent approaches consider unknown changes in subpopulation proportions and adapt calibration thresholds or weights to preserve prediction validity under such shifts \citep{wang2025conformalpredictionadaptiveunknown}.

These methods primarily address changes within a specified or estimable shift model. The question considered in this paper is different. We do not assume that the unseen source can be expressed through a known density ratio, subpopulation mixture, or predetermined shift family. Instead, we first ask whether the observed stream is still generated by any of the calibrated sources. This changes the problem from adapting a prediction threshold under a known form of shift to detecting when the set of calibrated sources itself is insufficient.

\paragraph{Online conformal prediction.}
Online conformal methods address sequential prediction when the data distribution may evolve over time. Adaptive conformal inference updates miscoverage thresholds as observations arrive \citep{gibbs2021adaptive}, while related methods extend these ideas to time series and other nonstationary settings \citep{zaffran2022adaptive}. Decaying step-size procedures can improve stability when the underlying stream is approximately stationary \citep{pmlr-v235-angelopoulos24a}. Online conformal ideas have also been applied to downstream objectives in LLM systems, including fairness-aware decision rules \citep{fayyazi25a}.

Our objective differs from online threshold adaptation. We model the stream as segments, each generated by a single source (Section~\ref{sec:formulation}), and test whether the current distribution belongs to the finite collection
\[
\{P_1,\ldots,P_K\}.
\]
Therefore, the main online quantity is a test statistic for this hypothesis, not an adaptively updated threshold for outcome prediction.

\paragraph{Jailbreak attacks and evolving attack distributions.}
Jailbreak attacks attempt to induce policy-violating behavior from aligned language models through adversarial prompting strategies \citep{zou2023universal,liu2023autodan}. Attack mechanisms include optimization-based suffixes, automated black-box search, and iterative prompt generation. Benchmarking frameworks such as HarmBench and JailbreakBench provide standardized protocols for evaluating attack success and guardrail robustness \citep{mazeika2024harmbench,chao2024jailbreakbench}.

A central challenge is that jailbreak strategies evolve over time. Longitudinal evaluations demonstrate that detection systems can experience substantial degradation under attack distribution shift \citep{piet2025jailbreaksovertime}. Recent guardrail surveys similarly emphasize that performance varies across attacks and that generalization to new strategies remains an important deployment challenge \citep{wang2025sokevaluatingjailbreakguardrails}. This motivates treating the attack source itself as a statistical object that may change during deployment. \citet{piet2025jailbreaksovertime} also propose two stream-level defenses. Their continuous detector is retrained on its own predictions, and their active monitor recognizes jailbreaks by whether they make the model answer known harmful prompts. Our detector uses different information. Specifically, it keeps the source calibration frozen and uses victim-model representations before any response is returned. Thus, it tests whether a recent segment is explained by any known source. Our benchmark replay does not compare against their procedures, so we make no claim of empirical superiority over them.

\paragraph{Representation-space analysis of jailbreak attacks.}
Several recent studies analyze jailbreak behavior through internal LLM representations, not only through prompt text. \citet{lin2024understandingjailbreakattacksllms} show that harmful and benign prompts can occupy distinct regions of hidden representation space in aligned language models. Subsequent work further studies safety boundaries in representation space \citep{gao2025shaping} and analyzes jailbreak behavior through model circuits and refusal-related pathways \citep{he2024jailbreaklens}. Representation-based safeguarding methods also manipulate internal model signals to steer generation toward safer behavior \citep{zheng2024prompt}.

These observations motivate our use of the victim-model hidden representation $X_t$ as the detection signal. The detector does not require the jailbreak outcome during deployment. The score measures how typical $X_t$ is under attack source $k$. It does not measure whether the current prompt is itself a successful jailbreak.

\paragraph{Prompt-level detection and stream-level detection of unseen sources.}
Most jailbreak guardrails are naturally formulated as prompt-level or response-level classifiers. Their objective is to determine whether a specific prompt, response, or interaction should be considered unsafe. Our problem is related but statistically distinct.

We consider
\[
H_0:Q\in\{P_1,\ldots,P_K\},
\]
where $Q$ denotes the distribution generating the current segment. The goal is therefore not to assign a harmfulness label to one observation. Instead, we determine whether the stream is still generated by at least one known attack source.

This distinction also clarifies the role of our method in a deployed safety system. A prompt-level guardrail can continue to screen individual inputs, while the proposed detector signals when the distribution of incoming attack representations has changed enough to warrant escalation, enrollment of a new source, or recalibration. Therefore, the two mechanisms can operate together as complements, which is the design evaluated in Section~\ref{sec:defense}.

\paragraph{Relation to e-detectors.}
The closest sequential change-detection literature deserves a detailed comparison. \citet{shin2022detectors} construct e-detectors by aggregating e-processes started at candidate change points and analyze false-alarm control in terms of the average run length (ARL). \citet{saha2026non} develop non-partitioned e-detectors over broad composite distribution classes with guarantees on the ARL or the probability of false alarm (PFA) under suitable weighting. Those constructions already handle composite nulls by taking an infimum over the candidate null distributions. Thus, aggregating over the $K$ known sources is not by itself a new device. However, our null hypothesis differs in what it contains, namely a set of calibrated sources. Thus, a switch from $P_i$ to $P_j$ is not a detection target when both are known. The window-based statistic is therefore designed for finite forgetting and for testing whether one known source explains a whole recent window, not for accumulating evidence for an arbitrary change point. Table~\ref{tab:positioning} summarizes the differences.

\begin{table}[H]
\centering
\caption{Comparison with the cited common-law change-detection formulation. The entries describe the stated null hypotheses and do not claim a limitation of every e-detector construction. ``KK switch'' is a switch between two known sources; ARL is the average run length and PFA the probability of false alarm.}
\label{tab:positioning}
\footnotesize
\begin{tabular}{p{0.20\textwidth}p{0.34\textwidth}p{0.36\textwidth}}
\toprule
Aspect & Cited common-law formulation & Persist-$L$ in this work \\
\midrule
Null hypothesis & one common no-change law from a specified class & the current segment follows one known source \\
KK switch & departs from the common-law null & not a target; a decision that mixes two known sources can still alarm \\
Composite structure & candidate null laws can be combined by an infimum & source minimum is taken after window averaging \\
Memory & cumulative, or restarted at candidate change points & recent window (Persist-3) or cumulative (SR-DKW) \\
Main guarantee & ARL or PFA control, depending on the construction & source-wise conformal validity, finite forgetting, fixed-window power \\
Role of e-processes & main detection statistic & auxiliary check only (Appendix~\ref{app:auxiliary}) \\
\bottomrule
\end{tabular}
\end{table}

\paragraph{Sequential hypothesis testing and e-processes.}
Repeated detection introduces an additional statistical challenge. A fixed-time hypothesis test does not automatically retain its validity when the practitioner checks a statistic at every step and stops whenever sufficient evidence is observed. Nonnegative supermartingales and e-processes provide a natural framework for controlling this form of repeated testing.

For completeness, Appendix~\ref{app:auxiliary} records a DKW-adjusted product statistic that has crossing control under a fixed known source, conditional on a high-probability calibration event. This statistic is not one of our detectors, does not appear in the comparisons, and is empirically very conservative. SR-DKW and SR-RAW are separate cumulative recursions and are not covered by that result. The multi-source structure nevertheless introduces the same logical requirement:
\[
H_0=\bigcup_{k=1}^{K}H_{0,k},
\;
H_{0,k}:Q=P_k.
\]
Rejecting a single $H_{0,k}$ does not establish an unseen source. The product statistic instead accumulates, at each step, the largest $p$-value over sources; it does not run separate source-wise sequential tests.

The main paper therefore does not use the product statistic as evidence of better performance. The central comparison is between the window-based Persist-3 and the cumulative SR detectors under the same calibration on known sources.

\paragraph{Positioning of the proposed framework.}
The proposed framework combines ideas that have usually been studied separately. Multi-source conformal methods provide uncertainty guarantees across heterogeneous calibrated distributions. Online and sequential methods provide tools for repeated testing over a stream. Jailbreak detection methods exploit model outputs or internal representations to identify unsafe or anomalous behavior.

Our contribution lies at their intersection. We use victim-model representations to construct source-wise conformal evidence, formulate the detection of unseen sources as a test of a union of source hypotheses, use a recent window for the finite-memory detector, and connect a detection event to stream-level escalation on top of prompt-level guardrails and, under additional conditions, to fresh enrollment of the new source and downstream multi-distribution conformal prediction.

\section{Implementation Details of Persist-3}
\label{app:algorithm}

\begin{figure*}[t]
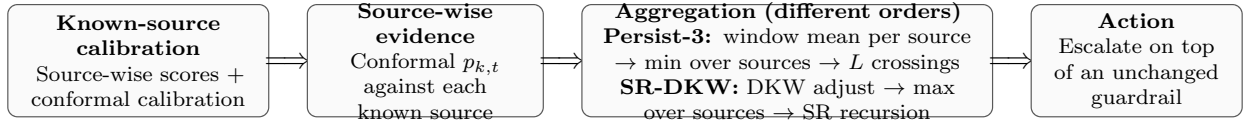

\centering
\tcbset{
  pipelinebox/.style={
    colback=black!2,colframe=black!65,boxrule=0.45pt,arc=1.2mm,
    left=1.0mm,right=1.0mm,top=0.8mm,bottom=0.8mm,
    height=15mm,valign=center,halign=center,fontupper=\scriptsize,
    before skip=0pt,after skip=0pt
  }
}
\begin{minipage}[c]{0.21\textwidth}
  \begin{tcolorbox}[pipelinebox]
    \textbf{Known-source calibration}\par
    Source-wise scores + conformal calibration
  \end{tcolorbox}
\end{minipage}%
\makebox[0.03\textwidth][c]{$\Longrightarrow$}%
\begin{minipage}[c]{0.19\textwidth}
  \begin{tcolorbox}[pipelinebox]
    \textbf{Source-wise evidence}\par
    Conformal $p_{k,t}$ against each known source
  \end{tcolorbox}
\end{minipage}%
\makebox[0.03\textwidth][c]{$\Longrightarrow$}%
\begin{minipage}[c]{0.33\textwidth}
  \begin{tcolorbox}[pipelinebox]
    \textbf{Aggregation (different orders)}\par
    \textbf{Persist-3:} window mean per source $\to$ $\min$ over sources $\to$ $L$ crossings\par
    \textbf{SR-DKW:} DKW adjust $\to$ $\max$ over sources $\to$ SR recursion
  \end{tcolorbox}
\end{minipage}%
\makebox[0.03\textwidth][c]{$\Longrightarrow$}%
\begin{minipage}[c]{0.17\textwidth}
  \begin{tcolorbox}[pipelinebox]
    \textbf{Action}\par
    Escalate on top of an unchanged guardrail
  \end{tcolorbox}
\end{minipage}
\caption{Overview of the pipeline. The two detectors combine the same source-wise conformal $p$-values in different orders. Persist-3 averages the evidence against each known source over a recent window and only then takes the minimum over sources. Thus, one known source must explain the whole window (Section~\ref{sec:detectnew}). In contrast, SR-DKW takes the largest DKW-adjusted $p$-value over sources at each step and accumulates over the whole history (Appendix~\ref{app:auxiliary}). An alarm activates escalation on top of an unchanged guardrail, which keeps screening every prompt. Realized alarm rates are in Table~\ref{tab:matched-main}.}
\label{fig:pipeline}
\end{figure*}

\paragraph{Precedent for the persistence rule.}
Requiring $L$ consecutive crossings is the special case $M=N=L$ of the $M$-out-of-$N$ rule in radar detection, which declares a detection when at least $M$ of $N$ successive tests exceed a threshold \citep{schwartz1956coincidence}. It is also closely related to the runs rules used with Shewhart control charts \citep{champ1987exact}.

\begin{algorithm}[t]
\caption{Persist-$L$ (Persist-3 uses $W=10$, $L=3$)}
\label{alg:detectnew}
\begin{algorithmic}[1]
\STATE \textbf{Input:} known-source training and calibration data; stream $\{X_t\}$; window length $W$; persistence length $L$; threshold $\widehat\tau$ selected by known-source validation for the protocol in use (values in the Calibration paragraph below)
\STATE Fit a shared PCA map $\Pi$ on pooled known-source training data, with dimension at most 32
\FOR{$k=1,\ldots,K$}
    \STATE Fit $\widehat\mu_k$ and Ledoit--Wolf $\widehat\Sigma_k$ on $\Pi(X)$ from source $k$
    \STATE Store $\mathcal R_k=\{S_k(X_i):X_i\in\mathcal D_{\mathrm{cal}}^{(k)}\}$
\ENDFOR
\FOR{$t=1,2,\ldots$}
    \STATE Observe $X_t$ and set $z_t\gets\Pi(X_t)$
    \FOR{$k=1,\ldots,K$}
        \STATE $S_{k,t}\gets(z_t-\widehat\mu_k)^\top\widehat\Sigma_k^+(z_t-\widehat\mu_k)$
        \STATE $p_{k,t}\gets
        \dfrac{1+\sum_{r\in\mathcal R_k}\Ind\{r\ge S_{k,t}\}}
        {|\mathcal R_k|+1}$
        \STATE $C_{k,t}\gets$ mean of $-\log p_{k,r}$ over $r=\max(1,t-W+1),\ldots,t$
    \ENDFOR
    \STATE $C_t^{\mathrm{fam}}\gets\min_k C_{k,t}$
    \IF{$t\ge L$ and $C_s^{\mathrm{fam}}\ge\widehat\tau$ for all $s=t-L+1,\ldots,t$}
        \STATE \textbf{return} $(\texttt{true},t)$
    \ENDIF
\ENDFOR
\STATE \textbf{return} $(\texttt{false},\mathrm{NA})$
\end{algorithmic}
\end{algorithm}

\paragraph{Calibration.}
In the short-horizon protocol, the caps of Section~\ref{sec:calibration} are set from the validation alarm rates $\mathrm{FPR}^{\mathrm{single}}_{\mathrm{val}}$ and $\mathrm{KK}^{\mathrm{single}}_{\mathrm{val}}$ of the single-crossing rule:
\[
\mathrm{FPR}_{\mathrm{cap}}
=
\max\{0.025,\mathrm{FPR}^{\mathrm{single}}_{\mathrm{val}}\},
\quad
\mathrm{KK}_{\mathrm{cap}}
=
\max\{0.05,\mathrm{KK}^{\mathrm{single}}_{\mathrm{val}}\}.
\]
For the frozen experiment, these caps are $0.031313$ and $0.051678$. With $W=10$ and $L=3$ fixed, and $\mathcal T$ the prespecified grid of 1183 candidate thresholds,
$
\widehat{\tau}
=
\min\left\{
\tau\in\mathcal T:
\widehat{\mathrm{FPR}}_{\mathrm{val}}(\tau)\le 0.031313,
\;
\widehat{\mathrm{KK}}_{\mathrm{val}}(\tau)\le 0.051678
\right\},
$
which gives $\widehat{\tau}\approx 2.386$ on the original split, with validation FPR $0.00502$ and validation KK alarm rate $0.04980$. For the robustness resplit (Appendix~\ref{app:maintext-deferred}), the same rule gives $\widehat\tau=2.374$. The long-horizon protocol of Section~\ref{sec:matched-regimes} uses its own caps (stationary alarm $\le0.01$, KK alarm $\le0.05$). The selected thresholds are $2.556$ for Persist-3, $68.69$ for SR-DKW full, $2.84\times10^{39}$ for SR-RAW full, $27.12$ for SR-DKW with $W=10$, and $141.7$ for SR-RAW with $W=10$. Since these thresholds live on method-specific scales, they are not comparable across methods. No KU data are used to choose $W$, $L$, or $\tau$. Persist-3 uses the raw split-conformal $p$-values; it does not use the DKW correction or the betting functions of Appendix~\ref{app:auxiliary}. Since the thresholds of different methods live on method-specific scales, we report realized alarm counts in Table~\ref{tab:matched-main} and do not call the methods exactly FPR-matched.

\section{Historical Known-to-Known Traces}
\label{app:kk-cases}

Figure~\ref{fig:kk-cases} shows traces of the family statistic $C_t^{\mathrm{fam}}$ on KK streams under a historical protocol, in which the threshold is set per configuration at the 99th percentile ($Q_{99}$) of validation values. These panels are retained only as finite-window recovery illustrations and are not frozen-protocol Persist-3 case studies. The frozen-protocol evidence for finite forgetting is the path-identity check and the alarm-timing result of Appendix~\ref{app:kk-timing}, not this figure.

\begin{figure}[H]
\centering
\includegraphics[width=\linewidth]{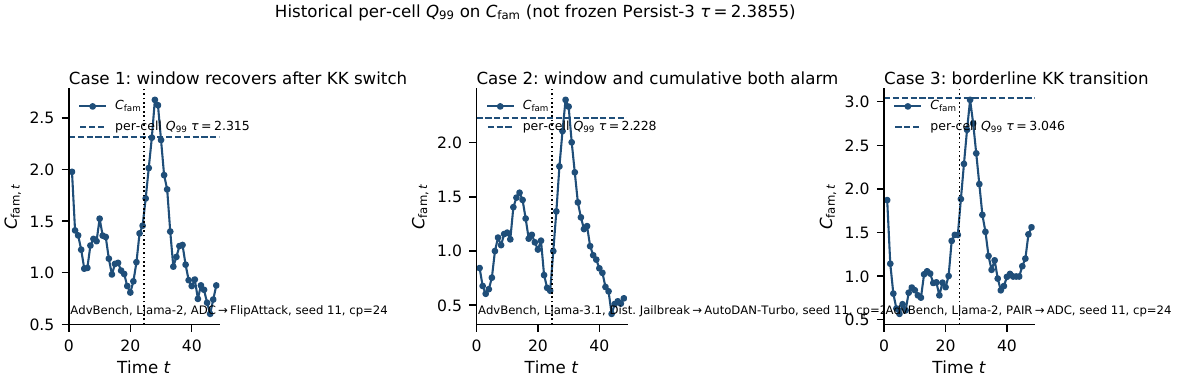}
\caption{Traces of the family statistic $C_t^{\mathrm{fam}}$ on KK streams under a historical per-configuration $Q_{99}$ threshold protocol. The X-axis is the stream position and the Y-axis is the family statistic. These are historical illustrations: the thresholds shown are per-configuration $Q_{99}$ values, not the frozen $\widehat\tau$ of Appendix~\ref{app:algorithm}, and the panels are not Persist-3 case studies under the frozen protocol.}
\label{fig:kk-cases}
\end{figure}

\section{DKW-Adjusted Product and SR Detectors}
\label{app:auxiliary}

This section defines the DKW-adjusted quantities used by SR-DKW and gives a validity result for a related product statistic. The result is a sanity check only: it does not apply to Persist-3, SR-DKW, or SR-RAW, and it gives no control under KK switches.

Let
\[
\varepsilon_k(\beta)
=
\sqrt{\frac{\log(2K/\beta)}{2n_k}},
\;
\widetilde p_k(x)=\min\{1,p_k(x)+\varepsilon_k(\beta)\},
\;
\widetilde p_t^{\mathrm{fam}}
=
\max_{k\in[K]}\widetilde p_k(X_t),
\]
with $\beta=0.05$, and let $g_\lambda(p)=1+\lambda(1-2p)$ for $\lambda\in[0,1]$. SR-DKW updates $R_{t,\lambda}=(1+R_{t-1,\lambda})\,g_\lambda(\widetilde p_t^{\mathrm{fam}})$ with $R_{0,\lambda}=0$, and SR-RAW applies the same recursion to $p_t^{\mathrm{fam}}=\max_k p_k(X_t)$. Both are empirical comparators; we do not prove guarantees for them.

For the product statistic, fix a grid $\Lambda=\{\lambda_m\}_{m=1}^{M}\subset[0,1]$ and weights $w_m\ge0$ with $\sum_m w_m=1$ before the stream starts, and define
\[
E_t^{(m)}
=
\prod_{s=1}^{t}
g_{\lambda_m}(\widetilde p_s^{\mathrm{fam}}),
\;
E_t^{\mathrm{mix}}
=
\sum_m w_m E_t^{(m)}.
\]
Here the maximum over sources is taken at each step before the product over time.

\begin{lemma}[DKW-adjusted inclusive upper-tail validity]
\label{lem:conditional}
Use $F_k,\widehat F_k,\varepsilon_k$ from Lemma~\ref{lem:cal-tail}, and let $\mathcal G_k$ be the $k$th event in the intersection defining $\mathcal G$. Define $p_k^{\rm score}(s)=(1+\#\{i:R_{k,i}\ge s\})/(n_k+1)$ and $\widetilde p_k^{\rm score}(s)=\min\{1,p_k^{\rm score}(s)+\varepsilon_k\}$. On $\mathcal G_k$,
\[
\widetilde p_k^{\rm score}(s)\ge1-F_k(s^-),\; s\in\mathbb R.
\]
Consequently, for every fixed $\mathcal O$ in $\mathcal G_k$ and fresh $X\sim P_k$ independent of $\mathcal O$, $\widetilde p_k(X)=\widetilde p_k^{\rm score}(S_k(X))$ is super-uniform conditional on $\mathcal O$. No continuity assumption is required.
\end{lemma}

\begin{theorem}[Fixed-source product crossing control]
\label{thm:eprocess}
Under Assumption~\ref{ass:monitoring}, let $\mathcal G$ be the event in Lemma~\ref{lem:cal-tail}. Suppose the entire stream is i.i.d.\ from a fixed known source $P_j$. For every offline realization $\mathcal O$ in $\mathcal G$, each $E_t^{(m)}$ and $E_t^{\rm mix}$, initialized at one, is a nonnegative supermartingale conditional on $\mathcal O$, with filtration $\mathcal F_t=\sigma(\mathcal O,X_1,\ldots,X_t)$. Thus, for $\delta\in(0,1)$,
\[
\Prob_{P_j}\!\left\{\sup_{t\ge0}E_t^{\rm mix}\ge1/\delta\mid\mathcal O\right\}\le\delta
\;\text{on }\mathcal G.
\]
Averaging over offline data gives
\[
\Prob_{P_j}\!\left\{\sup_{t\ge0}E_t^{\rm mix}\ge1/\delta\right\}
\le\delta+(1-\delta)\Prob(\mathcal G^c)\le\delta+\beta.
\]
For $\delta=\beta=0.05$, the bound is $0.10$; observing zero alarms does not reduce it.
\end{theorem}

Theorem~\ref{thm:eprocess} holds conditional on the calibration event $\mathcal G$ and does not make $E_t^{\rm mix}$ an unconditional e-process. Empirically, at $\alpha=0.05$, the same product with raw $p$-values alarms on about $0.184$ of stationary streams, and for about $11.9\%$ of (configuration, $\lambda$) pairs the empirical mean of $g_\lambda$ exceeds $1.01$. The DKW-adjusted product raises no alarm on stationary streams and has an empirical mean of $g_\lambda$ of about $0.781$, so it is very conservative.

\section{Proofs}
\label{app:proofs}

\subsection{Proof of Proposition~\ref{prop:source_validity}}
\begin{proof}
Condition on all training information $\mathcal T$, including the pooled data used to fit the shared PCA map. The $n_k$ calibration scores and one fresh score are exchangeable. Rank the $n_k+1$ scores in descending order, i.e., assign rank 1 to the largest score, and break ties uniformly at random. Then the rank $R$ of the fresh score is uniform on $\{1,\ldots,n_k+1\}$. The inclusive upper-tail conformal value satisfies $p_k(X)\ge R/(n_k+1)$. Therefore
\[
\Prob\{p_k(X)\le u\mid\mathcal T\}\le\frac{\lfloor(n_k+1)u\rfloor}{n_k+1}\le u.
\]
With no ties, equality with the normalized rank gives the discrete-uniform statement. Averaging over $\mathcal T$ gives marginal validity; conditioning additionally on the realized calibration set need not preserve uniformity.
\end{proof}

\subsection{Proof of Lemma~\ref{lem:cal-tail}}
\begin{proof}
Fix $\mathcal T$. Assumption~\ref{ass:monitoring} makes the scores within each calibration sample i.i.d.\ with CDF $F_k$. The DKW inequality gives
\[
\Prob\!\left\{\sup_s|\widehat F_k(s)-F_k(s)|>\varepsilon_k\mid\mathcal T\right\}
\le2e^{-2n_k\varepsilon_k^2}=\beta/K.
\]
A union bound gives $\Prob(\mathcal G\mid\mathcal T)\ge1-\beta$, without requiring independence across sources. Fix an offline realization in $\mathcal G$ and write
\[
\widehat{\overline F}_k(s)=1-\widehat F_k(s^-),\;
\overline F_k(s)=1-F_k(s^-),\;
p_k^{\rm score}(s)=\frac{1+n_k\widehat{\overline F}_k(s)}{n_k+1}.
\]
The DKW bound also holds at left limits, and
\[
p_k^{\rm score}(s)\ge\widehat{\overline F}_k(s)\ge\overline F_k(s)-\varepsilon_k.
\]
Thus $p_k(X)\le u$ implies $\overline F_k(S_k(X))\le u+\varepsilon_k$. For a fresh draw from $P_k$, $\overline F_k(S_k(X))$ is super-uniform, including when scores have atoms. Independence of the fresh draw from $\mathcal O$ therefore gives the conditional tail bound. Since $0\le Z_k\le b_k$, integrating this tail yields
\[
\E[Z_k(X)\mid\mathcal O]=\int_0^{b_k}\Prob\{Z_k(X)>z\mid\mathcal O\}\,dz
\le\int_0^{b_k}\min\{1,e^{-z}+\varepsilon_k\}\,dz.
\]
\end{proof}

\subsection{Proof of Theorem~\ref{thm:persist-null}}
\begin{proof}
Suppose the true source is $P_j$. Because $C_s^{\rm fam}=\min_k C_{k,s}$, a Persist-$L$ alarm at time $t$ implies
\[
C_{j,s}\ge\tau,
\; s=t-L+1,\ldots,t.
\]
Each source-$j$ statistic is an average over $W$ observations, so
\[
\sum_{s=t-L+1}^{t}W C_{j,s}\ge LW\tau.
\]
The left-hand side is a weighted sum of the $B=W+L-1$ observations entering the overlapping windows:
\[
\sum_{s=t-L+1}^{t}W C_{j,s}
=
\sum_{h=1}^{B}a_h Z_{j,h},
\]
where $a_h$ is the number of the $L$ windows containing observation $h$. Because $\sum_h a_h=LW$ and $\E[Z_{j,h}\mid\mathcal O]=\mu_{0,j}$, the alarm event implies
\[
\sum_{h=1}^{B}a_h(Z_{j,h}-\mu_{0,j})
\ge
LW(\tau-\mu_{0,j}).
\]
Conditional on $\mathcal O$, the $Z_{j,h}$ are independent and lie in $[0,b_j]$. The weighted Hoeffding inequality therefore yields
\[
\Prob\{D_t=1\mid\mathcal O\}
\le
\exp\left\{
-\frac{2L^2W^2(\tau-\mu_{0,j})^2}
{b_j^2\sum_h a_h^2}
\right\},
\]
which is Eq.~\eqref{eq:persist-null-bound}. For $W\ge L$, the overlap counts are
\[
1,2,\ldots,L-1,
\underbrace{L,\ldots,L}_{W-L+1\text{ terms}},
L-1,\ldots,2,1,
\]
which gives Eq.~\eqref{eq:awl}.
\end{proof}

\subsection{Proof of Corollary~\ref{cor:dkw-null}}
\begin{proof}
On the event of Lemma~\ref{lem:cal-tail}, $\mu_{0,j}=\E_{P_j}[Z_j(X)\mid\mathcal O]\le\bar\mu_j$. If $\tau>\bar\mu_j$, then $\tau-\mu_{0,j}\ge\tau-\bar\mu_j>0$. Substituting this smaller gap into Theorem~\ref{thm:persist-null} gives the bound.
\end{proof}

\subsection{Proof of Corollary~\ref{cor:finite-horizon}}
\begin{proof}
For each $t\ge L$, an alarm implies $C_{j,s}\ge\tau$ at all $L$ required times. Summing the window averages, including their startup denominators, gives
\[
\sum_{s=t-L+1}^t C_{j,s}=\sum_{r=1}^t v_{t,r}Z_j(X_r)\ge L\tau,
\; \sum_{r=1}^t v_{t,r}=L.
\]
Conditional on $\mathcal O$, the summands depend on independent observations and lie in intervals of length $v_{t,r}b_j$. The weighted Hoeffding inequality gives $\rho_{j,t}$ when the mean gap is positive; otherwise the probability is at most one. A union bound over $\{D_t=1\}$ for $L\le t\le H$ gives the bound without independence across overlapping windows. For full windows, the nonzero weights are $a_h/W$, so $V_t=A_{W,L}/W^2$. On $\mathcal G$, replace the mean by its upper bound as in Corollary~\ref{cor:dkw-null}. If the resulting sum is at most $\alpha$ on every offline realization in $\mathcal G$, integrate over $\mathcal O$ and add $\Prob(\mathcal G^c)\le\beta$.
\end{proof}

\subsection{Proof of Proposition~\ref{prop:kk-recovery}}
\begin{proof}
At time $t\ge B$, the earliest observation entering any of the $L$ windows is $X_{t-B+1}$. All other inputs to the decision are frozen offline objects. Thus, $D_t$ is a function of $X_{t-B+1},\ldots,X_t$, and two streams sharing this block have the same decision. If $t\ge\nu+B-1$, this block excludes every observation before $\nu$. If the block also lies in a known-source segment from $P_j$, its observations are i.i.d.\ conditional on $\mathcal O$ by Assumption~\ref{ass:monitoring}. Then Theorem~\ref{thm:persist-null} applies when $\tau>\mu_{0,j}$. These statements concern the continued statistic, not the action taken after the first alarm.
\end{proof}

\subsection{Proof of Proposition~\ref{prop:threshold-relax}}
\begin{proof}
A Persist-$L$ alarm in a stream requires $L$ consecutive threshold crossings, which contain at least one crossing. Therefore $A_e^{(L)}(\tau)\le A_e^{(1)}(\tau)$ for every stream $e$ and threshold $\tau$, where $A_e^{(L)}(\tau)$ indicates a Persist-$L$ alarm. Averaging over any fixed validation collection gives $\widehat r_L(\tau)\le\widehat r_1(\tau)$ for the empirical alarm rates. Hence every threshold feasible for single crossing under given caps is also feasible for Persist-$L$.
\end{proof}

\subsection{Proof of Theorem~\ref{thm:persist-power}}
\begin{proof}
For each source $k$ and each full post-change window,
\[
\E_Q Z_k(X)=\tau+\Delta_k(Q),
\;
0\le Z_k(X)\le b_k.
\]
Hoeffding's lower-tail inequality gives
\[
\Prob_Q\{C_{k,s}<\tau\mid\mathcal O\}
\le
\exp\left(-\frac{2W\Delta_k^2(Q)}{b_k^2}\right).
\]
The event $C_s^{\rm fam}\ge\tau$ fails only if at least one of the $K$ source statistics is below $\tau$, so a union bound over $k$ gives
\[
\Prob_Q\{C_s^{\rm fam}<\tau\mid\mathcal O\}
\le
\sum_{k=1}^{K}
\exp\left(-\frac{2W\Delta_k^2(Q)}{b_k^2}\right).
\]
Persist-$L$ fails at time $t$ if at least one of its $L$ required crossings fails. A second union bound over these $L$ times gives Eq.~\eqref{eq:persist-power}. No independence among the overlapping windows is required.
\end{proof}

\subsection{Proof of Corollary~\ref{cor:nearest-source}}
\begin{proof}
Because $\Delta_k(Q)\ge\Delta_*(Q)$ and $b_k\le b_{\max}$ for every $k$,
\[
\exp\left(-\frac{2W\Delta_k^2(Q)}{b_k^2}\right)
\le
\exp\left(-\frac{2W\Delta_*^2(Q)}{b_{\max}^2}\right).
\]
Substituting this bound into Theorem~\ref{thm:persist-power} yields Eq.~\eqref{eq:nearest-source-power}.
\end{proof}

\subsection{Proof of Corollary~\ref{cor:delay}}
\begin{proof}
Fix $\mathcal O$ and consider $t_m=\nu+mB-1$ for integers $m\ge1$ with $mB\le H$. The decision $D_{t_m}$ uses a block contained in $[\nu+(m-1)B,\nu+mB-1]$. These blocks are disjoint and independent under the i.i.d.\ segment assumption, and each failure probability is at most $q$ by Theorem~\ref{thm:persist-power}. If $T_{\nu,H}>mB$, the first $m$ block decisions must all fail, so
\[
\Prob_Q\{T_{\nu,H}>mB\mid\mathcal O\}\le\prod_{\ell=1}^m\Prob_Q\{D_{t_\ell}=0\mid\mathcal O\}\le q^m.
\]
The case $m=0$ is trivial. For finite $H$, no eligible alarm implies failure at the $\lfloor H/B\rfloor$ block endpoints, giving the miss bound. For $H=\infty$, summing the delay tail in groups of $B$ terms gives
\[
\E_Q[T_{\nu,\infty}\mid\mathcal O]
=\sum_{d=0}^{\infty}\Prob_Q\{T_{\nu,\infty}>d\mid\mathcal O\}
\le B\sum_{m=0}^{\infty}q^m=\frac{B}{1-q}.
\]
Only fully post-change decisions enter, so mixed windows need not be controlled.
\end{proof}

\subsection{Proof of Proposition~\ref{prop:identifiability}}
\begin{proof}
Fix $\mathcal O$ and the common initial state. By independence across observations, the law of the score process is the product of its one-observation laws. Equality of these laws under $Q$ and $P_j$ therefore gives equality of the entire score-process laws, and hence of any measurable alarm-time law. The same argument applies conditional on a common history when both future streams are independent of it. For dependent streams, equality of one-observation laws alone would not suffice.
\end{proof}

\subsection{Proof of Lemma~\ref{lem:conditional}}
\begin{proof}
For a fixed offline realization in $\mathcal G_k$, the argument in Lemma~\ref{lem:cal-tail} gives
\[
p_k^{\rm score}(s)\ge1-\widehat F_k(s^-)\ge1-F_k(s^-)-\varepsilon_k.
\]
Adding $\varepsilon_k$ and truncating at one preserves domination of $1-F_k(s^-)\in[0,1]$. Since $1-F_k(S_k(X)^-)$ is super-uniform for a fresh draw, so is the inflated value, including for score distributions with atoms.
\end{proof}

\subsection{Proof of Theorem~\ref{thm:eprocess}}
\begin{proof}
Fix $\mathcal O$ in $\mathcal G$. The fresh observation $X_t\sim P_j$ is independent of the past conditional on $\mathcal O$. Lemma~\ref{lem:conditional} gives conditional super-uniformity of $\widetilde p_j(X_t)$, and $\widetilde p_t^{\rm fam}\ge\widetilde p_j(X_t)$. Because $g_\lambda$ is non-increasing for $0\le\lambda\le1$,
\[
\E[g_\lambda(\widetilde p_t^{\rm fam})\mid\mathcal F_{t-1}]
\le\E[g_\lambda(\widetilde p_j(X_t))\mid\mathcal F_{t-1}]
\le\int_0^1[1+\lambda(1-2u)]\,du=1.
\]
Each factor is nonnegative, so each product starting at one is a nonnegative supermartingale, and so is their convex combination with fixed weights. Ville's inequality gives the conditional crossing bound. Integrating over $\mathcal G$ and using probability at most one on $\mathcal G^c$ gives $\delta+(1-\delta)\Prob(\mathcal G^c)\le\delta+\beta$.
\end{proof}

\subsection{Proof of Corollary~\ref{cor:post_enrollment_mdcp}}
\begin{proof}
Fix an enrollment history $\mathcal H$ satisfying the assumptions and a true source $j$. Conditional exchangeability of its calibration and future labeled scores gives $\Prob_j\{\pi_j(X,Y)\le\alpha\mid\mathcal H\}\le\alpha$ by the rank argument in Proposition~\ref{prop:source_validity}. Since $\max_k\pi_k(X,Y)\ge\pi_j(X,Y)$, exclusion from $\Gamma_\alpha(X)$ implies $\pi_j(X,Y)\le\alpha$. This proves coverage for each source without independence across calibration blocks. If the enrollment decision depended on an old calibration set, exchangeability for that set cannot be inferred from its original marginal validity, and a fresh block for the new source does not repair this.
\end{proof}

\section{Additional Operating Points}
\label{app:latest-operating-audits}

\subsection{Short-horizon comparison}
Before the main comparison in Section~\ref{sec:matched-regimes}, we compared methods on short streams of 40 observations per segment, with looser caps: validation stationary FPR $\le0.05$ and validation KK alarm rate $\le0.05$. For each method, the most aggressive threshold satisfying both caps is chosen on known-source validation streams only.

We write OUT-WIN for the aggregation used by Persist-3: average $-\log p_k$ of each source over a window of $W=10$ observations, and then take the minimum over sources, which gives $C_t^{\mathrm{fam}}$ in \eqref{eq:family-stat}. The methods in Table~\ref{tab:short-horizon} are as follows.
\begin{itemize}
\item \emph{Persist-3 OUT-WIN-$W10$} is Persist-3 (Algorithm~\ref{alg:detectnew}).
\item \emph{Single-crossing OUT-WIN-$W10$} uses the same $C_t^{\mathrm{fam}}$ but alarms at the first crossing, i.e., $L=1$.
\item \emph{Mean-$p$ conformal window} averages the source-wise conformal $p$-values, $p_t^{\mathrm{mean}}=K^{-1}\sum_{k}p_{k,t}$, and uses the mean of $-\log p_t^{\mathrm{mean}}$ over the last $W=10$ observations. It does not take the minimum over sources.
\item \emph{Family min-kNN window} works in the PCA space without conformal $p$-values. For each source $k$, it uses the average distance to the $q$ nearest calibration points of that source, $S_k^{\mathrm{knn}}(x)=q^{-1}\sum_{j=1}^{q}\|\Pi(x)-\Pi(x^{(k)}_{(j)})\|_2$, where $x^{(k)}_{(j)}$ is the $j$th nearest one and $q\in\{5,10,20\}$ is chosen once per configuration by the area under the ROC curve (AUROC) on known-source validation data. It then averages $S_k^{\mathrm{knn}}$ over the last $W=10$ observations and takes the minimum over sources.
\end{itemize}

\begin{table}[H]
\centering
\caption{Short-horizon comparison on the 330 configurations. ``Val'' and ``Test'' denote validation and test streams; FPR is the stationary alarm rate, KK the alarm rate on streams with a known-to-known switch, and KU the detection rate on streams with a known-to-unknown switch. Mean delay is over detected streams.}
\label{tab:short-horizon}
\footnotesize
\begin{tabular}{lcccccc}
\toprule
Method & Val FPR & Val KK & Test FPR & Test KK & Test KU & Mean delay\\
\midrule
Persist-3 OUT-WIN-$W10$ & 0.0050 & 0.0495 & 0.0000 & 0.0349 & 0.4663 & 8.94\\
Single-crossing OUT-WIN-$W10$ & 0.0466 & 0.0439 & 0.0256 & 0.0494 & 0.3755 & 7.45\\
SR-RAW & 0.0500 & 0.0017 & 0.0207 & 0.0477 & 0.7990 & 13.18\\
SR-DKW & 0.0491 & 0.0027 & 0.0134 & 0.0215 & 0.6926 & 12.19\\
Mean-$p$ conformal window & 0.0466 & 0.0001 & 0.0313 & 0.0001 & 0.2148 & 12.36\\
Family min-kNN window & 0.0027 & 0.0481 & 0.0046 & 0.0463 & 0.1569 & 7.07\\
\bottomrule
\end{tabular}
\end{table}

We also evaluated two variants that use the Mahalanobis score without conformal calibration. \emph{Family min-Mahalanobis window} uses $S_{k,t}^{M}=(Z_t-\widehat\mu_k)^\top\widehat\Sigma_k^{+}(Z_t-\widehat\mu_k)$ with $Z_t=\Pi(X_t)$, averages it over the last $m_t$ observations, and takes the minimum over sources, $C_t^{M}=\min_k m_t^{-1}\sum_{u=t-m_t+1}^{t}S_{k,u}^{M}$; \emph{Persist-3 on raw Mahalanobis window} requires three consecutive crossings of $C_t^{M}$. Neither raw-Mahalanobis variant had a feasible operating point within the prespecified threshold grid under the known-only validation caps, so neither has an operating point to report. The grids were $n=1169$ candidates with $\tau\in[3.00,413.41]$ for the min-Mahalanobis window and $n=1174$ candidates with $\tau\in[3.00,166.52]$ for Persist-3 on the raw Mahalanobis window. This is a statement about the grid, not an impossibility: on a finite validation set, a threshold above every observed statistic never alarms and would trivially satisfy both caps, but such thresholds lie outside the prespecified grid.

CUSUM satisfies the caps on validation streams but has a test KK alarm rate of $0.991$, because short validation streams do not reproduce the state that CUSUM accumulates over long KK streams. After matching validation and test lengths, CUSUM has KU detection $0$ (Appendix~\ref{app:followup-abd}), so we do not use it as a comparator.

\subsection{Window length}
\begin{table}[H]
\centering
\caption{Persist-3 with different window lengths $W$, under the short-horizon calibration of the previous subsection.}
\footnotesize
\begin{tabular}{rccccc}
\toprule
$W$ & $\tau$ & Test FPR & Test KK & Test KU & Mean delay\\
\midrule
3 & 2.2627 & 0.0227 & 0.0209 & 0.5228 & 11.80\\
5 & 2.1263 & 0.0191 & 0.0593 & 0.5723 & 9.52\\
10 & 2.3861 & 0.0000 & 0.0349 & 0.4663 & 8.94\\
\bottomrule
\end{tabular}
\end{table}
Neither $W=3$ nor $W=5$ is better than $W=10$ in all of FPR, KK alarm rate, and KU detection. We keep $W=10$, which was fixed in advance for its low stationary alarm rate, and do not select $W$ by KU results. Detectability also varies across victim models and attack sources (Figure~\ref{fig:sourcewise-ku-latest}). We report these differences and do not select models.

\begin{figure}[H]
\centering
\adjustimage{trim=0 0 0 {0.10\height},clip,width=0.8\linewidth}{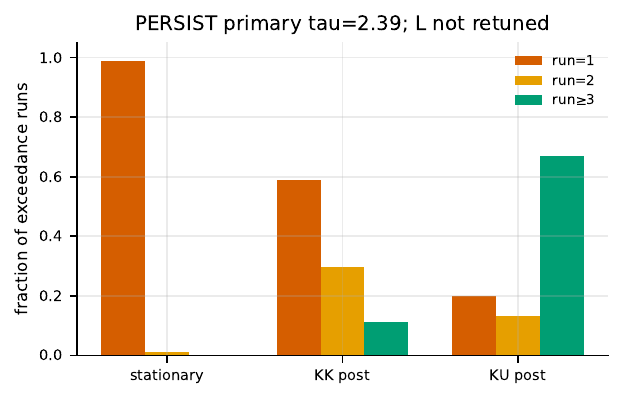}
\caption{Run-length distribution of threshold exceedances of the pre-persistence family statistic $C_t^{\mathrm{fam}}$, on the original 330-configuration split for Persist-3 with $W=10$, $L=3$ (not retuned) and $\widehat\tau=2.386$. A run is a maximal block of consecutive positions with $C_t^{\mathrm{fam}}\ge\widehat\tau$. The X-axis separates stationary streams, post-switch KK segments and post-switch KU segments; the Y-axis is the fraction of exceedance runs of length 1, 2, or at least 3. Runs are taken over full stationary streams and over the post-switch part of KK and KU streams. Fractions are normalized over exceedance runs, not over streams, so a stream with no exceedance contributes no run: the counts are 79 stationary runs, 4471 post-switch KK runs and 1622 post-switch KU runs.}
\label{fig:persist3-run-lengths}
\end{figure}

Figure~\ref{fig:persist3-run-lengths} is the empirical counterpart of the persistence rule: Persist-3 fires exactly when such a run reaches length $L=3$, so the mass at length $\ge3$ under known sources is what the threshold has to suppress, and the mass at length $\ge3$ on post-switch KU segments is what it has to keep.

\subsection{Persistence length}
\label{app:persist-length}
This subsection examines how the persistence length $L$ affects the detector. We rerun all 330 original-split configurations under the short-horizon protocol, with $W=10$ and the PCA dimension $32$ fixed. For each $L\in\{1,2,3,4\}$, the threshold $\tau$ is recalibrated by the known-only rule of Section~\ref{sec:calibration}, under the same caps as Persist-3 ($0.031313$ for stationary alarms and $0.051678$ for KK alarms). No unseen-source data are used. Rates are source-macro averages with 95\% intervals from a bootstrap over held-out sources (5000 replicates, seed 11).

\begin{table}[H]
\centering
\caption{Detectors with different persistence lengths $L$ (330 configurations, short-horizon protocol, $W=10$). $\tau$ is recalibrated for each $L$ on known-source validation streams. $L=3$ is the primary choice, fixed in advance. Delay is the mean delay over detected streams.}
\label{tab:persist-length}
\footnotesize
\begin{tabular}{rccccc}
\toprule
$L$ & $\tau$ & FPR $\downarrow$ & KK $\downarrow$ & KU $\uparrow$ & Delay \\
\midrule
1 & 2.997 & $0.016\,[0.014,0.018]$ & $0.000\,[0.000,0.001]$ & $0.207\,[0.139,0.301]$ & 7.78 \\
2 & 2.541 & $0.000\,[0.000,0.000]$ & $0.035\,[0.024,0.040]$ & $0.428\,[0.319,0.579]$ & 8.62 \\
3 (primary) & 2.386 & $0.000\,[0.000,0.000]$ & $0.035\,[0.024,0.040]$ & $0.467\,[0.352,0.623]$ & 8.94 \\
4 & 2.234 & $0.000\,[0.000,0.000]$ & $0.039\,[0.027,0.045]$ & $0.521\,[0.384,0.693]$ & 9.99 \\
\bottomrule
\end{tabular}
\end{table}

Table~\ref{tab:persist-length} gives three observations. First, persistence is necessary. With $L=1$, even the high recalibrated threshold leaves stationary false alarms ($0.016$), and it gives a low KU detection rate ($0.207$). Second, for $L\ge2$, the stationary FPR is zero and the KK alarm rate stays between $0.035$ and $0.039$. Third, KU detection increases with $L$ ($0.428\to0.467\to0.521$), while the mean delay also grows ($8.62\to8.94\to9.99$). Thus, $L$ trades detection against delay. Note that $L=3$ is not claimed to be optimal. It was fixed in advance by the known-only pilot, and we do not reselect it using KU results. Finally, the recalibrated $L=1$ operating point differs from the single-crossing row of Table~\ref{tab:short-horizon}, which uses a different threshold rule, so the two should not be compared directly.

\subsection{PCA dimension}
\label{app:pca-dim}

This subsection examines how the PCA dimension $d$ of the shared map $\Pi$ affects Persist-3. We rerun all 330 original-split configurations from the cached embeddings under the short-horizon protocol, with $W=10$ and $L=3$ fixed. For each $d$, the threshold $\tau$ is recalibrated by the same known-only rule of Section~\ref{sec:calibration}, and no unseen-source data are used. Rates are source-macro averages with 95\% intervals, as in Table~\ref{tab:pseudo-confirmatory-app}.

\begin{table}[H]
\centering
\caption{Persist-3 with different PCA dimensions $d$ (330 configurations, short-horizon protocol). $\tau$ is recalibrated for each $d$ on known-source validation streams. KK is the alarm rate on streams with a known-to-known switch, KU the detection rate on streams with a known-to-unknown switch, and Delay the mean delay over detected streams.}
\label{tab:pca-dim}
\footnotesize
\begin{tabular}{rcccc}
\toprule
$d$ & $\tau$ & KK $\downarrow$ & KU $\uparrow$ & Delay \\
\midrule
8  & 2.380 & $0.029\,[0.022,0.034]$ & $0.353\,[0.251,0.522]$ & 8.32 \\
16 & 2.386 & $0.040\,[0.027,0.047]$ & $0.417\,[0.316,0.563]$ & 9.28 \\
32 (primary) & 2.386 & $0.035\,[0.024,0.040]$ & $0.467\,[0.352,0.623]$ & 8.28 \\
64 & 2.392 & $0.036\,[0.024,0.042]$ & $0.465\,[0.335,0.611]$ & 8.30 \\
\bottomrule
\end{tabular}
\end{table}

Table~\ref{tab:pca-dim} gives two observations. First, the KK alarm rate is insensitive to $d$, staying between $0.029$ and $0.040$. Second, KU detection increases from $d=8$ to $d=32$ ($0.353\to0.417\to0.467$) and then levels off ($0.465$ at $d=64$). Thus, a very low dimension discards information that separates the unseen source from the known ones, while dimensions beyond 32 bring no further gain. Note that the intervals overlap across dimensions, so these differences are descriptive. The primary choice $d=32$ was fixed before this ablation and was not selected by KU results.

\subsection{Robustness resplit}
\begin{table}[H]
\centering
\caption{Original split and robustness resplit (Appendix~\ref{app:maintext-deferred}) under the short-horizon protocol, reported separately and never pooled. $N$ denotes the number of streams. Delay is omitted because the robustness resplit did not store it in the same form.}
\label{tab:pseudo-confirmatory-app}
\footnotesize
\begin{tabular}{lcccccc}
\toprule
Evaluation & FPR & KK & KU & $N_{\rm stat}$ & $N_{\rm KK}$ & $N_{\rm KU}$\\
\midrule
Original split & $0/1835$ & 0.035 & 0.467 & 1835 & 10640 & 1835\\
Robustness resplit & $10/2425$ & $0.044\,[0.033,0.048]$ & $0.414\,[0.282,0.589]$ & 2425 & 17040 & 2425\\
\bottomrule
\end{tabular}
\end{table}
The 95\% Clopper--Pearson upper bound for the robustness-resplit stationary alarm rate is $0.0070$.

\subsection{Lower alarm budgets, memory, and repeated switches}
\label{app:followup-abd}

We also ran a follow-up protocol on the same 330 configurations, fixed before KU results were summarized. It uses a different construction of KK validation streams and does not replace the comparisons above.

\paragraph{Lower alarm budgets.}
In this follow-up protocol, the validation cap on stationary alarms is $0.005$, with a KK cap of $0.05$. We also consider stationary caps in $\{0.001,0.01,0.025,0.05\}$. At $0.005$, SR-DKW stays within the cap on test streams and has KU detection $0.6295$, versus $0.4516$ for Persist-3; SR-RAW reaches $0.7857$ but has test FPR $0.0087$, above the cap. At $0.001$, neither Persist-3 nor SR-DKW raises a stationary alarm, while SR-DKW has a lower KK alarm rate ($0.0018$ vs.\ $0.0036$) and higher KU detection ($0.580$ vs.\ $0.320$).

\paragraph{Order of aggregation and memory.}
With the same $p$-values, $W$, $L$, and calibration, OUT-WIN (average over the window, then minimum over sources, as in Persist-3) gives KU detection $0.452$, versus $0.400$ for IN-WIN, which takes the minimum over sources at each step and then averages over the window. Within the SR family, full SR-RAW has KU detection $0.786$, versus $0.630$ when truncated to $W=10$, so truncation reduces power.

\paragraph{Dependence on observations before a switch.}
The paired-stream check of finite forgetting under this protocol is reported in Appendix~\ref{app:kk-timing}. In addition, with thresholds fixed at the $0.005$ cap, the stationary alarm rates at horizons $H=40,200,1000$ are $(0.0005,0.003,0.010)$ for Persist-3 and $(0,0.035,0.148)$ for SR-DKW. Horizons longer than the test pool are generated by sampling with replacement.

\paragraph{Repeated switches.}
When the stream switches among known sources every $d$ steps, Persist-3 raises an alarm on a fraction $0.61$ of streams at $d=5$, $0.21$ at $d=12$, and $0.027$ at $d=40$. When the dwell time is smaller than $B$, no decision block may lie entirely in one segment. Thus, the guarantees for homogeneous blocks do not cover these decisions. Note that short segments alone do not contradict the piecewise i.i.d.\ model. The relevant restriction is whether all observations used by a decision come from one source. Thus, these results are stress tests only.

\subsection{Dwell time shorter than the decision block}
\label{app:dwell-time}

This subsection examines how Persist-3 behaves when the attacker switches faster than its decision block of $B=W+L-1=12$ observations. We call the number of consecutive prompts from one source the \emph{dwell time} $D$. We use the cached streams of the 50 configurations in Appendix~\ref{app:defense}, with a shared horizon of $T=328$. The detector is frozen, with $W=10$, $L=3$, and $\widehat\tau=2.386$, and nothing is retuned. We consider two regimes. In the unknown$\to$unknown regime, consecutive segments of length $D$ all carry unseen-source behavior. In the known$\to$unknown reset regime, unseen-source segments alternate with known-source segments, each of length $D$. \emph{Segment detection} is the fraction of unseen-source segments with an alarm, and \emph{episode detection} is the fraction of streams with at least one alarm on an unseen-source segment.

\begin{table}[H]
\centering
\caption{Persist-3 under different dwell times $D$ with the frozen detector ($W=10$, $L=3$, $B=12$).}
\label{tab:dwell-time}
\footnotesize
\begin{tabular}{rcccc}
\toprule
& \multicolumn{2}{c}{Unknown$\to$unknown} & \multicolumn{2}{c}{Known$\to$unknown reset} \\
\cmidrule(lr){2-3}\cmidrule(lr){4-5}
$D$ & Segment detection & Episode detection & Segment detection & Episode detection \\
\midrule
3  & 0.34 & 0.75 & 0.08 & 0.61 \\
12 & 0.55 & 0.77 & 0.45 & 0.68 \\
24 & 0.57 & 0.77 & 0.46 & 0.65 \\
\bottomrule
\end{tabular}
\end{table}

Table~\ref{tab:dwell-time} gives two observations. First, under unknown$\to$unknown switching, detection degrades gracefully. Segment detection falls from $0.55$ at $D=12$ to $0.34$ at $D=3$, while episode detection stays around $0.75$. Second, under known$\to$unknown resetting, fast switching hurts much more. Segment detection falls from $0.45$ at $D=12$ to $0.08$ at $D=3$, since the known-source segments repeatedly pull the window back toward a known source. However, episode detection remains $0.61$ at $D=3$, so fast switching does not lead to complete evasion in this test. Thus, $B$ is a temporal resolution of the theorems, not a requirement for the detector to run. Note that a fully adaptive attacker, which observes the detector state and switches to destroy the evidence margin, is outside both our theorems and this stress test.

\section{Theory--Experiment Diagnostics}
\label{app:similarity}

The following diagnostics use the original split under the short-horizon protocol and are computed after the detector is frozen; they are not used for tuning. Source-level associations use the eleven held-out attack sources, with 5000 bootstrap replicates (seed 11).

\subsection{Nearest-source separation and KU detection}
For a held-out source $Q$ with $n_Q$ test observations, the estimated nearest-source separation is
$
\widehat\Delta_*(Q)
=
\min_{k\in[K]}
\{\widehat\mu_k(Q)-\tau\},
\;
\widehat\mu_k(Q)
=
\frac{1}{n_Q}\sum_{i=1}^{n_Q} -\log p_k(X_i).
$
It is strongly associated with KU detection: Pearson $r=0.98$ $[0.94,1.00]$ and Spearman $\rho=0.95$ $[0.70,1.00]$; the robustness resplit gives Pearson $r=0.981$ $[0.947,0.995]$. Thus, sources that stay close to at least one known source are hard to detect, consistent with Corollary~\ref{cor:nearest-source} and Proposition~\ref{prop:identifiability}. These are observational associations, not independent validations of the theory.

\begin{figure}[H]
\centering
\includegraphics[width=0.7\linewidth]{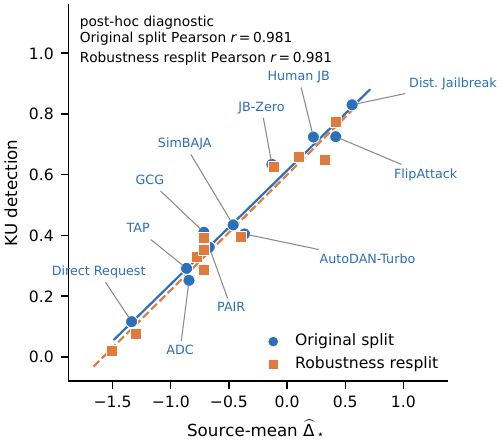}
\caption{Source-level relation between the estimated nearest-source separation $\widehat\Delta_*$ and KU detection, for the original split and the robustness resplit. The robustness resplit re-partitions the same cached embedding pool, so the two series are not independent samples of jailbreak behavior.}
\label{fig:delta-star-diagnostics}
\end{figure}

\subsection{Detection delay}
Among detected streams, the median delay of each source is positively associated with $\widehat\Delta_*$ (Pearson $r=0.84$ $[0.70,0.96]$, Spearman $\rho=0.80$; robustness resplit $r=0.866$). This association is driven by selection: sources with small $\widehat\Delta_*$ are rarely detected, and when they are, it tends to be by early chance crossings. We therefore do not treat it as a test of Corollary~\ref{cor:delay}. Since a delay summary over detected streams is a selected summary, we report the detection rate and the delay among detected streams separately. Coding missed detections as the horizon plus one mixes detection and delay, so we do not use it to validate Corollary~\ref{cor:delay}.

\subsection{Empirical checks of finite forgetting}
\label{app:kk-timing}
\label{app:finite-forgetting-check}
This subsection gives the two empirical checks of finite forgetting summarized in Section~\ref{sec:matched-regimes}. With $W=10$ and $L=3$, Proposition~\ref{prop:kk-recovery} states that the decision of Persist-3 cannot depend on any pre-switch observation after $B=12$ post-switch observations.

\paragraph{Paired-stream check.}
This check uses no threshold. Under the follow-up protocol of Appendix~\ref{app:followup-abd}, we run each detector on two streams that differ before a switch but are identical after it. In all 330 configurations, Persist-3 makes identical decisions on the two streams from $B=12$ steps after the switch. In contrast, full SR keeps a different state (median absolute difference $3.01$). Thus, the finite memory of Persist-3 holds exactly in the implementation, while the cumulative state of SR keeps reflecting pre-switch evidence.

\paragraph{Timing of KK alarms.}
Throughout this paragraph, the delay of an alarm at time $t$ is the number of post-switch observations it has seen, $d=t-\nu+1$, which the saved runs store as $t-\mathrm{cp}$. In this indexing, Proposition~\ref{prop:kk-recovery} excludes every pre-switch observation once $d\ge B=12$. Thus, pre-switch observations can influence only the decisions with $d\le11$. Any alarm at $d\ge12$ uses observations from the new known source alone, so it is an ordinary stationary alarm, not an effect of the old source. Finite forgetting removes the old source. However, it does not rule out later alarms.

The timing of first KK alarms is reported for the frozen primary detector on the original split only: Persist-3 with $W=10$, $L=3$ and $\widehat\tau=2.386$. Among the $N_{\rm KK}=10640$ KK episodes of that split, 507 have a first post-switch alarm, and all $507/507$ occur at $d\le11$.

Two points should be noted. First, the count 507 is the number of episodes with a first post-switch alarm. It is neither the source-macro KK rate nor the pooled any-alarm rate, so $507/10640\approx0.0477$ is not comparable with the KK alarm rate of $0.035$ in Table~\ref{tab:pseudo-confirmatory-app}, which is computed over a different event and aggregation. Second, the corresponding robustness-resplit timing uses the separately recalibrated $\widehat\tau=2.374$, so it belongs to a different protocol and is not reported here or shown in any frozen-protocol figure.

\begin{figure}[H]
\centering
\begin{subfigure}[t]{0.48\textwidth}
    \centering
    \includegraphics[width=0.7\linewidth]{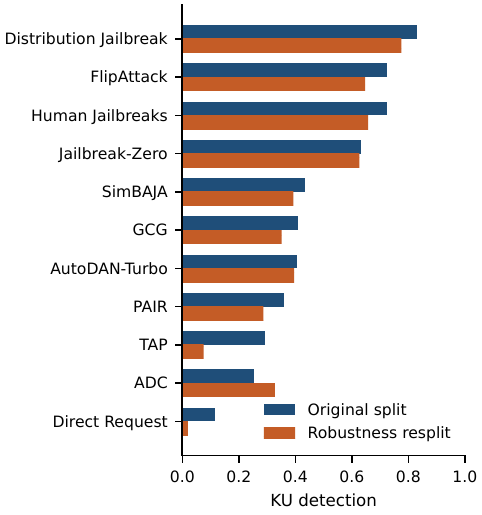}
    \caption{KU detection for each of the eleven held-out attack sources, on the original split and in the robustness resplit.}
    \label{fig:sourcewise-ku-latest}
\end{subfigure}
\hfill
\begin{subfigure}[t]{0.48\textwidth}
    \centering
    \includegraphics[width=0.7\linewidth]{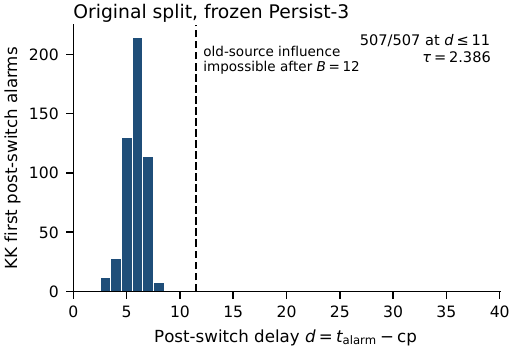}
    \caption{Timing of the first post-switch Persist-3 alarm on KK test streams of the original split, under $W=10$, $L=3$ and frozen $\widehat\tau=2.386$. The X-axis is the number of post-switch observations the alarm has seen, $d=t_{\rm alarm}-\nu+1$, stored in the saved runs as $t_{\rm alarm}-\mathrm{cp}$. Among $N_{\rm KK}=10640$ episodes, 507 have a first post-switch alarm, and all $507/507$ occur at $d\le11$. The vertical line marks $B=W+L-1=12$, the first value of $d$ at which no pre-switch observation can enter the decision.}
    \label{fig:kk-alarm-timing-app}
\end{subfigure}
\caption{Left: KU detection by held-out source. Right: frozen-protocol timing of the first post-switch KK alarm on the original split. The robustness-resplit timing uses a separately recalibrated threshold and is not shown here.}
\label{fig:robustness-forgetting-diagnostics}
\end{figure}

\begin{table}[H]
\centering
\caption{Diagnostics corresponding to the theoretical results for Persist-3.}
\label{tab:theory-empirical}
\footnotesize
\begin{tabular}{p{0.20\textwidth}p{0.25\textwidth}p{0.25\textwidth}p{0.22\textwidth}}
\toprule
Result & Diagnostic & Original split & Robustness resplit \\
\midrule
Nearest-source separation & $\widehat\Delta_*$ vs.\ KU detection & Pearson $0.98$; Spearman $0.95$ & Pearson $0.981\,[0.947,0.995]$ \\
Finite forgetting & paired streams agree after $B=12$ steps & 330/330 configurations & n/a \\
First post-switch KK alarm & timing within 11 steps, frozen $\widehat\tau$ & 507/507 of $N_{\rm KK}=10640$ & separate protocol, not reported \\
Delay & $\widehat\Delta_*$ vs.\ delay among detected & Pearson $0.84$ & Pearson $0.866$ \\
\bottomrule
\end{tabular}
\end{table}

\section{Comparison of the Two Detectors}
\label{app:method-families}

Persist-3 and SR-DKW differ in two choices. For combining sources, Persist-3 averages each source over the window and then takes the minimum over sources, while SR-DKW takes the largest DKW-adjusted $p$-value over sources at each step. For combining over time, Persist-3 uses a recent window with persistence, while SR-DKW uses a cumulative Shiryaev--Roberts recursion. The SR recursion and the DKW inequality are standard; we use them to compare cumulative and finite memory under the same calibration. Empirically, SR-DKW detects more unseen sources, while the decision of Persist-3 depends only on the most recent $B=W+L-1$ observations.

\begin{remark}[Finite memory versus cumulative SR detection]
\label{rem:sr-memory}
Persist-3 discards observations older than $B$ steps, whereas SR-DKW keeps a cumulative state unless it is truncated or reset. This structural difference is not a power advantage: under matched long-horizon calibration, SR-DKW attains higher KU detection with comparable observed alarm rates. Memory is also not the only difference between the two detectors, which additionally differ in the order of source aggregation and in the DKW adjustment. Thus, the cross-method gap should not be read as a memory effect on its own, and the within-family truncation ablation in Section~\ref{sec:matched-regimes} is the cleaner comparison. SR-DKW is preferred when sensitivity is the priority, while Persist-3 is preferred when the decision must depend only on recent observations. Rapid switching or interleaving of sources is not covered by the guarantees.
\end{remark}

\begin{table}[H]
\centering
\caption{Which detector to use.}
\footnotesize
\begin{tabular}{p{0.45\textwidth}p{0.18\textwidth}p{0.30\textwidth}}
\toprule
Requirement & Prefer & Reason \\
\midrule
Maximize KU detection & SR-DKW & Weak evidence accumulates over long histories \\
Decision depends only on recent observations & Persist-3 & Decision uses the last $B$ observations only \\
Remove the effect of the old source from the continued decision after a switch & Persist-3 & Removed after $B=W+L-1$ new observations \\
Rapid switching or interleaving of sources & Neither by default & Not covered by the guarantees \\
\bottomrule
\end{tabular}
\end{table}

\section{Data}
\label{app:data}

The eleven attack sources are ADC~\citep{hu2024efficient}, AutoDAN-Turbo~\citep{liu2025autodan}, Direct Request~\citep{mazeika2024harmbench}, Distribution Jailbreak~\citep{fu2026steering}, FlipAttack~\citep{liu2024flipattack}, GCG~\citep{zou2023universal}, Human Jailbreaks~\citep{shen2024doanythingnow}, Jailbreak-Zero~\citep{hu2026jailbreak}, PAIR~\citep{chao2025jailbreaking}, SimBAja~\citep{takemoto2024all}, and TAP~\citep{mehrotra2024tree}. The benchmarks contain harmful requests and jailbreak prompts. We use only victim-model representations and binary jailbreak labels, and the paper does not display any harmful prompt or response.

A \emph{configuration} is a combination of dataset, victim model, held-out source, and seed. For the source-detection experiments (Sections~\ref{sec:matched-regimes} and \ref{sec:separation} and Appendices~\ref{app:maintext-deferred}, \ref{app:latest-operating-audits}, and \ref{app:similarity}), we ran 67 leave-one-source-out jobs over 5 seeds, giving 335 configurations; the 5 configurations for AdvBench, Gemma, and held-out Distribution Jailbreak failed, which leaves 330 configurations. Table~\ref{tab:configs} lists them by dataset and victim model. The layered evaluation of Section~\ref{sec:defense} uses a separate registry of 50 configurations (Appendix~\ref{app:defense}).

\begin{table}[H]
\centering
\caption{Number of completed configurations in the source-detection experiments, by dataset and victim model.}
\label{tab:configs}
\footnotesize
\begin{tabular}{llc}
\toprule
Dataset & Victim model & Configurations \\
\midrule
AdvBench & Llama-2 & 55 \\
AdvBench & Qwen & 50 \\
AdvBench & Gemma & 50 \\
AdvBench & Llama-3.1-8B-Instruct & 55 \\
HarmBench & Llama-2 & 30 \\
HarmBench & Qwen & 30 \\
HarmBench & Gemma & 30 \\
HarmBench & Llama-3.1-8B-Instruct & 30 \\
\midrule
Total & & 330 \\
\bottomrule
\end{tabular}
\end{table}

The held-out source contributes no training, calibration, or validation data and appears only in the test streams. Under the short-horizon protocol, test streams sample 40 observations per segment from the test pool without replacement within a stream; the same prompt may appear in different streams. Confidence intervals for the original split use a bootstrap over held-out sources with 5000 replicates (seed 11).

The original split was designed after earlier KU results had been seen, so its evaluation is exploratory. The robustness resplit (Appendix~\ref{app:maintext-deferred}) uses the same selection rule, which gives $\widehat\tau=2.374$, and contains 2425 stationary, 17040 KK, and 2425 KU streams, with no failed runs. The two evaluations are never pooled.

\section{Layered Intervention Details}
\label{app:defense}

This section documents the layered evaluation of Section~\ref{sec:defense}. All numerical values in this section are the reported frozen-threshold snapshot; no recalibration under a common validation protocol is reported.

\subsection{Evaluation registry and data lineage}
\label{app:defense-registry}

The evaluation replays cached victim-model representations and original outcome labels on aligned streams. That is, the guardrail decision $G_t$, the monitor statistics, and the label $Y_t$ refer to the same prompt occurrence.

\paragraph{Configurations.}
All 50 configurations are AdvBench with the victim model Llama-2-7B-chat, indexed by held-out source and seed over seeds $\{0,7,11,23,42\}$ and the ten held-out sources ADC, AutoDAN-Turbo, Direct Request, Distribution Jailbreak, FlipAttack, GCG, Human Jailbreaks, Jailbreak-Zero, PAIR and SimBAja. The number of known sources is $K=7$ on 35 configurations and $K=8$ on 15 (Direct Request, Distribution Jailbreak and Human Jailbreaks, at all five seeds). They were selected by a gated first pass of six seed-11 configurations (ADC, AutoDAN-Turbo, FlipAttack, GCG, PAIR, SimBAja), which had to be complete, followed by expansion in lexicographic order over the remaining (held-out, seed) pairs until 50 complete caches were reached; TAP is next in that order and was never reached, so it does not appear. Configurations without a complete cache were skipped, and no complete configuration was dropped for an unfavorable KU outcome. This scope is narrower than the four-model source-detection study of Sections~\ref{sec:matched-regimes} and \ref{sec:separation}.

\paragraph{Streams.}
Streams are built with the same constructor as the source-detection study, with pre-switch segment length $n_{\rm each}=40$ in all 50 configurations. The run contains 365 stationary streams of length 40 with no switch, 2310 KK streams of length 80 covering every ordered pair of distinct known sources, and 365 KU streams of length 80, one per known source into the held-out source. Each KK and KU stream has 40 pre-switch observations, so in the one-based indexing of Section~\ref{sec:formulation}, where $\nu$ is the position of the first observation from the new source, the switch is at $\nu=41$; the saved runs record it as the change point $\mathrm{cp}=40$, the number of pre-switch observations. We use $\mathrm{cp}$ only when quoting saved quantities and $\nu$ everywhere else. Stream counts vary by configuration with $K$; the totals above are the realized 50-configuration sums. Test identifiers are drawn without replacement within a stream but independently across streams, so the same test prompt can appear in many streams.

\paragraph{The 1,729 denominator.}
The figure 1,729 counts aligned KU-stream positions whose jailbreak label is 1, summed over the 50 configurations, with a per-configuration range of 18 to 77. It counts prompt occurrences, not distinct prompt identifiers, and includes pre-switch known-source successes. The number of distinct prompts behind it is unresolved: per-position sample identifiers were not persisted in the saved result files.

\paragraph{Alignment.}
For each KU stream position, we match the guardrail decision, the monitor statistic, and the jailbreak label by the attack method and the behavior index of the prompt. The guardrail decisions come from the saved MTK and Prompt Guard~2 predictions, and the labels come from the cached jailbreak-success records. All 29{,}200 KU stream positions have both guardrail decisions, so no position is dropped. Note that this check does not verify the label quality or the uniqueness of prompts.

\paragraph{Data overlap.}
The evaluation reuses the AdvBench and Llama-2-7B-chat subset of the original split, with seeds $\{0,7,11,23,42\}$. Thus, it is not an independent data collection. The Self-Risk control uses the same 50 configurations without reselection. In addition, the replay keeps all cached observations, including those on which the guardrail intervened. Thus, unlike a live system, the monitor can still observe a request after the guardrail has blocked it.

\subsection{Policy, activation timing, and estimands}
\label{app:defense-estimands}

The combined action is $I_t^{(G+M)}=G_t\vee A_t^{(M)}$ with $A_t^{(M)}=\Ind\{t\ge T_{\rm act}^{(M)}\}$ (Section~\ref{sec:post-detection}). Escalation activates at the alarm itself, so $T_{\rm act}^{(M)}=T_{\rm alarm}$ and the alarmed prompt is itself escalated. If the monitor alarms before the KU switch, escalation stays on for every later position, including the whole post-switch suffix; such episodes are included in the KU intervention and extra-block metrics, but they are not counted as KU detections, which require a first alarm after the change point. An escalated prompt is counted as intercepted. This assumes that escalation prevents the harmful response, and it does not test a human review service. In addition, later prompts, representations, and labels stay those of the cached stream after the escalation. Thus, the replay does not model an attacker who reacts to the intervention.

For configuration $c$, the evaluation index set $\mathcal E_c$ contains every aligned KU-stream position whose jailbreak label is 1, both before and after the switch. It is the whole KU episode, not the post-switch part alone, and not the subset of successes left unblocked by the guardrail, which is only the numerator of the increment below. The blocking rate and its increment are
\[
B_c(G+M)=\frac{\sum_{i\in\mathcal E_c}Y_i I_i^{(G+M)}}{\sum_{i\in\mathcal E_c}Y_i},
\qquad
\Delta B_c=B_c(G+M)-B_c(G)=\frac{\sum_{i\in\mathcal E_c}Y_i(1-G_i)A_i^{(M)}}{\sum_{i\in\mathcal E_c}Y_i},
\]
where the second equality holds because the action is exactly $G\vee A$. Every term in that numerator is nonnegative, so $\Delta B_c\ge0$ by construction: under a monotone OR policy, adding any trigger can only add interceptions, never remove them. A positive $\Delta B$ therefore shows that the trigger fires on streams that still contain successful attacks, which is an operationally meaningful quantity, but it cannot by itself show that a source-aware trigger is more valuable than any other trigger with the same alarm budget. The comparison that would show this is to recalibrate Persist-3, SR-DKW, and Self-Risk on the same known-source validation streams, under the same alarm caps and the same escalation policy, freeze them, and only then evaluate KU streams. That experiment is not reported here, and the Self-Risk control in Appendix~\ref{app:defense-selfrisk} does not substitute for it because its calibration differs (Appendix~\ref{app:defense-audit}). The reported macro estimator is
\[
\Delta B_{\rm macro}=\sum_{c\in\mathcal C_+}w_c\,\Delta B_c,\qquad \sum_{c\in\mathcal C_+}w_c=1,
\]
where $\mathcal C_+$ is the set of configurations with a nonzero success denominator. The weights are uniform over configurations, $w_c=1/50$; the macro estimator is not weighted by the 1,729 occurrences or by the number of known sources. For an intervention domain $\mathcal D_c$, such as the prompts of the KU streams, the intervention increment is
\[
\Delta I_c=\frac{1}{|\mathcal D_c|}\sum_{i\in\mathcal D_c}\bigl(I_i^{(G+M)}-G_i\bigr).
\]
The pooled yield is $100\,N_{\rm extra,successful}/N_{\rm extra,intervened}$, with numerator and denominator taken from the same domain; a zero denominator is reported as N/A, not as zero. Known-source attack traffic in this evaluation is attack traffic, not benign traffic, so the burdens below are not false positive rates on benign users.

\subsection{Frozen operating points and validation audit}
\label{app:defense-audit}

Table~\ref{tab:layered-audit} reports the alarm rates of each trigger on the present known-source validation streams. These are trigger alarm rates, not intervention rates of the combined system: MTK and PG2 may already intervene on known-source streams, which is not a false alarm of the trigger. Self-Risk is calibrated on these streams under the two caps of $0.05$, whereas the thresholds of Persist-3 and SR-DKW are inherited from earlier experiments. The KK validation alarm rate of $0.103$ for Persist-3 therefore exceeds the cap imposed on Self-Risk, and the comparison in Section~\ref{sec:defense} is not budget-matched.

\begin{table}[H]
\centering
\caption{Validation audit of the triggers in the layered evaluation. ``Known-stream alarm'' is the fraction of validation streams from a single known source with a trigger alarm; ``KK-stream alarm'' is the same fraction for validation streams with a known-to-known switch.}
\label{tab:layered-audit}
\footnotesize
\begin{tabular}{lccc}
\toprule
Trigger & Known-stream alarm & KK-stream alarm & Threshold \\
\midrule
Self-Risk, MTK & 0.0015 & 0.000 & 0.310 \\
Self-Risk, PG2 & 0.0009 & 0.001 & 0.998 \\
Persist-3 (inherited) & 0.006 & 0.103 & 2.386 \\
SR-DKW (inherited) & 0.047 & 0.003 & 19.701 \\
\bottomrule
\end{tabular}
\end{table}

\subsection{Self-Risk control}
\label{app:defense-selfrisk}

The Self-Risk control uses the guardrail's own risk score, $r_t^{\rm MTK}=-\mathrm{mtk\_score}_t$ for MTK and $r_t^{\rm PG2}=\mathrm{pg2\_prob\_malicious}_t$ for PG2. It averages $r_t$ over the most recent $W=10$ observations and alarms after $L=3$ consecutive threshold crossings, as in Persist-3; it uses neither source-wise conformal evidence nor held-out KU outcomes. During startup it averages over the available prefix, $m_t=\min(W,t)$, with no burn-in period, and the persistence rule leaves the statistic undefined until three observations have arrived, so the earliest possible alarm is at $t=3$. MTK Self-Risk makes no additional successful interception and no additional KU prompt intervention, so its yield is N/A. PG2 Self-Risk makes two additional KU prompt interventions and no additional successful interception, so its yield is zero. Because of the calibration mismatch in Table~\ref{tab:layered-audit}, this result does not rule out that a differently calibrated temporal escalation rule would also add protection.

\subsection{Paired results and uncertainty}
\label{app:defense-counts}

Table~\ref{tab:layered-counts} reports pooled counts. The pooled fraction and the macro average differ: for example, $163/1729\approx0.094$ for MTK$+$Persist-3, whereas the macro-averaged gain in Table~\ref{tab:layered} is $0.074$, because the pooled ratio weights every successful occurrence equally while the macro average weights configurations by $w_c$. The reported intervals are 95\% percentile clustered-bootstrap intervals with 5000 replicates and seed 11, resampling the ten held-out attack sources as clusters and keeping all configurations of a source together. For the difference between the two monitors, the per-configuration values of the SR-DKW minus Persist-3 extra-block gain give a macro point estimate of $0.026$ for both MTK and PG2, but no interval was computed for that paired difference, so we make no paired significance claim. The separate intervals against the guardrail alone overlap (for MTK, $0.074\,[0.033,0.118]$ versus $0.100\,[0.041,0.161]$) and do not establish a difference between the monitors. The intervention sets of the two monitors are not assumed to be nested.

\begin{table}[H]
\centering
\caption{Pooled counts in the layered evaluation. ``Intercepted'' is the number of originally successful attack occurrences intercepted by the system, out of 1,729; ``Additional'' is the increase relative to the guardrail alone; ``Extra KU interventions'' is the number of additional prompt interventions in streams with a known-to-unknown switch; ``Yield'' is additional successful interceptions per 100 extra KU interventions.}
\label{tab:layered-counts}
\footnotesize
\begin{tabular}{lcccc}
\toprule
System & Intercepted & Additional & Extra KU interventions & Yield \\
\midrule
MTK & 682 & 0 & 0 & N/A \\
MTK $+$ Persist-3 & 845 & 163 & 2,240 & 7.28 \\
MTK $+$ SR-DKW & 915 & 233 & 3,043 & 7.66 \\
PG2 & 501 & 0 & 0 & N/A \\
PG2 $+$ Persist-3 & 682 & 181 & 3,403 & 5.32 \\
PG2 $+$ SR-DKW & 764 & 263 & 4,513 & 5.83 \\
\bottomrule
\end{tabular}
\end{table}

The KU detection rates of the monitors in this evaluation are $0.640$ for Persist-3 and $0.851$ for SR-DKW; they are computed on the layered streams and are not comparable to Table~\ref{tab:matched-main}.

\subsection{Historical standalone evaluation}
\label{app:defense-historical}

An earlier version of this paper compared Persist-3 with prompt-level defenses as standalone policies. That evaluation used a different protocol: AdvBench with Llama-2, 5720 prompts, 301 successful attack prompts, 940 streams, 55 leave-one-source-out configurations, and five seeds. After the first alarm, all remaining prompts in the stream are escalated. We keep it here for completeness. It is not an independent confirmation of Section~\ref{sec:defense}, and it is not the baseline of the layered gains in Table~\ref{tab:layered}.

\begin{table}[H]
\centering
\caption{Historical standalone comparison (55 configurations). Conditional defense success is the fraction of previously successful attack prompts blocked by the intervention policy; it is not residual ASR or an ASR reduction. Episode intervention is the fraction of streams containing an unseen source in which the policy intervenes. First intervention is the median position of the first intervention. Known burden and KK burden are the fractions of streams from a single known source and of KK streams, respectively, in which the policy intervenes.}
\label{tab:defense-historical}
\footnotesize
\begin{tabular}{lccccc}
\toprule
Method &
Conditional defense success &
Episode intervention &
First intervention &
Known burden &
KK burden \\
\midrule
Persist-3 & 0.675 & 0.875 & 31 & 0.009 & 0.124 \\
MTK & 0.410 & 1.000 & 1 & 0.708 & 0.708 \\
Prompt Guard 2 & 0.304 & 1.000 & 2 & 0.458 & 0.458 \\
\bottomrule
\end{tabular}
\end{table}

There are two observations. First, as standalone policies, Persist-3 blocks a larger fraction of previously successful attack prompts than MTK and Prompt Guard 2 ($0.675$ versus $0.410$ and $0.304$) and intervenes less often on streams from a single known source ($0.009$ versus $0.708$ and $0.458$). Second, it acts much later (median first intervention at position 31, versus 1 and 2), and its KK burden of $0.124$ is non-negligible. The known burdens concern attack streams from known sources, not benign users, so this comparison does not establish that Persist-3 can replace a guardrail or that it has a low burden on benign traffic.

We also compared Persist-3 with the \emph{historical defense single-crossing OUT-WIN}, an earlier detector used in this evaluation. It uses the same statistic $C_t^{\mathrm{fam}}$, but a separate split of the defense data, a threshold $\tau_k$ set per configuration at the 99th percentile of validation values, a strict crossing rule $C_t^{\mathrm{fam}}>\tau_k$, and the same escalation policy after the alarm. Relative to it, Persist-3 increases conditional defense success by $0.234$ $[0.060,0.489]$ and changes the known burden by $-0.056$ $[-0.070,-0.044]$, but increases the KK burden by $0.100$ $[0.068,0.119]$. Thus, Persist-3 improves defense success and reduces interventions on known sources, but not every intervention burden.
\end{document}